\documentclass[a4paper,10pt,reqno]{article}

\usepackage{amssymb}
\usepackage{latexsym}
\usepackage{amsmath}
\usepackage{graphicx}
\usepackage{amsthm}
\usepackage{empheq}
\usepackage{bm}
\usepackage{booktabs}

\usepackage[dvipsnames]{xcolor}
\usepackage{pagecolor}
\usepackage{subcaption}

\usepackage{mathptmx,cite}

\usepackage[margin=4cm]{geometry}

\numberwithin{equation}{section}

\theoremstyle{definition}
\newtheorem{theorem}{Theorem}[section]
\newtheorem{corollary}[theorem]{Corollary}
\newtheorem{proposition}[theorem]{Proposition}
\newtheorem{definition}[theorem]{Definition}
\newtheorem{example}[theorem]{Example}
\newtheorem{notation}[theorem]{Notation}
\newtheorem{remark}[theorem]{Remark}
\newtheorem{lemma}[theorem]{Lemma}

\newcommand\qbin[3]{\left[\begin{matrix} #1 \\ #2 \end{matrix} \right]_{#3}}

\newcommand{\numberset}{\mathbb}

\newcommand{\R}{\numberset{R}}

\newcommand{\F}{\numberset{F}}

\newcommand{\Tr}{\textnormal{Tr}}
\newcommand{\mV}{\mathcal{V}}

\newcommand{\mC}{\mathcal{C}}

\newcommand{\mE}{\mathcal{E}}

\newcommand{\mD}{\mathcal{D}}
\newcommand{\mA}{\mathcal{A}}
\newcommand{\mU}{\mathcal{U}}
\newcommand{\mW}{\mathcal{W}}
\newcommand{\mF}{\mathcal{F}}
\newcommand{\rk}{\textnormal{rk}}
\newcommand{\mB}{\mathcal{B}}
\newcommand{\col}{\operatorname{column-space}}

\newcommand{\mat}{\F_q^{n \times m}}
\renewcommand{\longrightarrow}{\to}

\newcommand{\maxrk}{\operatorname{maxrk}}

\newcommand{\inn}{\textnormal{in}}
\newcommand{\vect}{\textnormal{vect}}

\newcommand{\drk}{d}

\newcommand{\colsp}{\textnormal{column-space}}
\newcommand{\rowsp}{\textnormal{row-space}}

\newcommand{\Wrk}{W}

\newcommand\bbq[1]{\bm{b}_q(#1)}

\newcommand*{\myproofname}{Proof of the claim}

\usepackage{hyperref}

\usepackage{authblk}
\title{\textbf{Rank-Metric Codes and Their Parameters}}
\author{Anina Gruica, Altan B. K\i l\i \c{c}, Alberto Ravagnani}
\date{}                    
\affil{Department of Mathematics and Computer Science \\ Eindhoven University of 
	Technology, the Netherlands}

\usepackage{setspace}
\begin{document}

	\maketitle
	
	\thispagestyle{empty}
	
	\begin{abstract}
		We present the theory of linear rank-metric codes
		from the point of view of their fundamental parameters. These are: the minimum rank distance, the rank distribution, the maximum rank, the covering radius, and the field size.
		The focus of this chapter is on the interplay 
		among these parameters and on their significance for the code's (combinatorial) structure.
		The results covered in this chapter span from the 
		theory of optimal codes and anticodes to
		very recent developments on the asymptotic density of MRD codes.
	\end{abstract}

	\medskip
	
	\tableofcontents
	
	\medskip

\section*{Introduction}

A rank-metric code is a linear space of matrices
with a given size over a finite field $\F_q$,
in which the rank of any nonzero matrix is bounded from below by a given integer $d$, called the \textit{minimum distance} of the code. Rank-metric codes 
were first studied in connection with the theory of association schemes by Delsarte in 1978~\cite{delsarte1978bilinear}.
More recently, they were rediscovered in connection with various applications  in~\cite{roth1991maximum,cooperstein1997external,gabidulin1985theory}.

In 2008, rank-metric codes were proposed as a solution to the problem of error amplification in communication networks~\cite{silva2008rank}. Since then, they have been a thriving research topic 
within coding theory, electrical engineering,
and more generally discrete mathematics.
The interplay between the theory of rank-metric codes and other areas of pure and applied mathematics is at the roots
of major research themes within coding theory, with contributions intersecting semifield theory and finite geometry~\cite{sheekey2016new,csajbok2017maximum,zini2021scattered,gruica2022rank,de2016algebraic,gluesing2020sparseness,sheekey2020new,lunardon2017kernels,sheekey2019binary}, 
enumerative and algebraic combinatorics~\cite{delsarte1978bilinear,schmidt2015symmetric,ravagnani2018duality,schmidt2018hermitian,byrne2020partition,gruica2020common,delsarte1975alternating,britz2020wei,shiromoto2019codes,ghorpade2020polymatroid,jurrius2016defining,qpolyconn,gorla2020rank}, 
algebra~\cite{liebhold2016automorphism,el2021valued},
complexity theory~\cite{roth1996tensor,byrne2021tensor,byrne2019tensor,couvreur2020hardness}, and
cryptology (for a survey on this vast application area
we refer the reader to~\cite[Section~3]{bookappl}
and to the references therein).
Recently, rank-metric codes started to appear also in contemporary coding theory textbooks; see for instance~\cite{gorla2018codes,gorlach,gabidulin2021rank,bookappl}. We refer in particular to~\cite{bookappl} and to the references therein for the applied aspects of the theory, which we do not treat here.

The goal of this chapter is to offer an introduction to the \textit{mathematical theory} of rank-metric codes, with a focus on their fundamental parameters and on the significance of these for the code's combinatorial structure. In contrast with most references on the subject, in this chapter we give considerable space to the field size as a fundamental code parameter,
following also the emerging interest 
of the research community in the density properties of error-correcting codes.

The rest of the chapter is organized as follows: In Section~\ref{sec:1} we define rank-metric codes and their fundamental parameters. We devote Section~\ref{sec:2} to the MacWilliams Identities for the rank metric. Section~\ref{sec:3} focuses on the minimum distance of a rank-metric code, which is the key parameter for error correction.
In Section~\ref{sec:4} we study the maximum rank (weight) in a rank-metric code, which yields the notion of anticode in the rank metric. The covering radius and the external distance of a code are the subject of Section~\ref{sec:5}. We finally devote Section~\ref{sec:6} to the size of the underlying field of a rank-metric code and to its significance for density properties.

\section{Rank-Metric Codes}
\label{sec:1}
	
Throughout this chapter, $q$ denotes a prime power and $\F_q$ is the finite field of $q$ elements. Standard references on finite fields are \cite{lidl1997finite, mullen2013handbook}. We work with integers $m \ge n \ge 1$ and denote by $\smash{\mat}$ the space of $n \times m$ matrices over $\F_q$. We let ``$\le$'' denote the inclusion of vector spaces, to emphasize the linear structure. All the dimensions in this chapter are computed over $\F_q$, unless otherwise stated.
	
We start by defining the objects we will study throughout this chapter.
	
\begin{definition}
A (\textbf{rank-metric}) \textbf{code} is an
$\F_q$-linear subspace $\mC \le \mat$. The \textbf{dimension} of $\mC$ is its dimension as an $\F_q$-linear space.
\end{definition}

We measure the distance between matrices by computing the rank of their difference.

\begin{definition}
The \textbf{rank distance} between matrices $X,Y \in \mat$ is $\drk(X,Y) = \rk(X-Y)$. 
\end{definition}

\begin{remark}
One can check that the pair $(\mat,\drk)$ is a metric space, meaning that the rank distance map $\drk: \mat \times \mat \longrightarrow \R$ satisfies the following properties:
\begin{enumerate}
    \item \label{p1}$\drk(X,Y) = 0$ if and only if $X=Y$, for all $X,Y \in \mat$;
    \item \label{p2}$\drk(X,Y) = \drk(Y,X)$ for all $X,Y \in \mat$;
    \item \label{p3}$\drk(X,Z) \le \drk(X,Y)+\drk(Y,Z)$ for all $X,Y,Z \in \mat$.
\end{enumerate}
The latter property is called the \textit{triangular inequality}. In order to see why it holds, observe that for all $A,B \in \mat$ we have 
\begin{equation}\label{eq:col}
    \col(A+B) \le \col(A) + \col(B).
\end{equation}
Taking dimensions on both sides of~\eqref{eq:col} we obtain
\begin{align*}
    \rk(A+B) \le \rk(A)+\rk(B)-\dim(\col(A) \cap \col(B)) \le \rk(A) + \rk(B),
\end{align*}
where $\col(M)$ denotes the subspace of $\F_q^n$ generated by the columns of the matrix~$M$.  
Finally, by setting $A=X-Y$ and $B=Y-Z$ in the inequality above, we obtain the triangular inequality (Property~\ref{p3} above).
\end{remark}

\begin{example} \label{ex:rkcode}
Let $q=m=n=2$. Let $\mC \le \F_2^{2 \times 2}$ be the code generated by the two matrices
\begin{align*}
    \begin{pmatrix}
    0 & 1 \\
    1 & 1
    \end{pmatrix}, 
    \begin{pmatrix}
    1 & 1 \\
    1 & 0
    \end{pmatrix}.
\end{align*}
We have $\dim(\mC) = 2$ and 
\begin{align*}
    \mC = \left\{ \begin{pmatrix}
    0 & 0 \\
    0 & 0
    \end{pmatrix}, \;
    \begin{pmatrix}
    1 & 0 \\
    0 & 1
    \end{pmatrix}, \; 
    \begin{pmatrix}
    1 & 1 \\
    1 & 0
    \end{pmatrix}, \;
    \begin{pmatrix} 
    0 & 1 \\
    1 & 1
    \end{pmatrix}\right\}.
\end{align*}
\end{example}

The focus of this chapter is on some of the fundamental parameters of a rank-metric code: minimum distance, rank distribution, maximum rank, covering radius, and field size. The latter parameter is simply $q$ and it plays a crucial role 
in determining the proportion of codes having good distance properties; see Section~\ref{sec:6}. The other key parameters are defined as follows.

\begin{definition}
Let $\mC \le \mat$ be a rank-metric code. 
\begin{itemize}
    \item The \textbf{minimum} (\textbf{rank}) \textbf{distance} of $\mC$ is 
    $\drk(\mC)=\min\{\rk(X) \mid X \in \mC, \; X \neq 0\}$,
    where we artificially set
    $\drk(\{0\}) = n+1$ by definition.
    \item For $i \in \{0,1,\dots,n\}$, let
    $\Wrk_i(\mC) = |\{X \in \mC \mid \rk(X) = i\}|$ be the number of rank $i$ matrices contained in $\mC$. Then the tuple $(\Wrk_i(\mC) \mid i \in \{0,\dots,n\})$ is the \textbf{rank distribution} of $\mC$.
    \item The integer $\operatorname{maxrk}(\mC) = \max\{\rk(X) \mid X \in \mC\}$ is the \textbf{maximum rank} of $\mC$.
    \item The \textbf{covering radius} of $\mC$ is $$\rho(\mC) = \min\{i \mid \textnormal{ for all } X \in \mat \; \exists \, Y  \in \mC  \textnormal{ with } \drk(X,Y) \le i\}.$$ 
\end{itemize} 
\end{definition}

As we will often see throughout the chapter, duality plays an important role in the theory of rank-metric codes. We start by introducing a notion of scalar product between matrices.

\begin{definition}
The \textbf{trace product} between matrices $X,Y \in \mat$ is $\langle X,Y \rangle = \Tr(XY^T)$, where $\Tr(A) = \sum_{i=1}^n A_{ii} \in \F_q$ is the \textbf{trace} of a matrix $A \in \F_q^{n \times n}$.
\end{definition}

\begin{remark} \label{rem:nond}
For $X,Y \in \mat$ we have $\langle X,Y \rangle = \sum_{i,j=1}^n X_{ij}Y_{ij}$. For example,
\begin{align*}
    \left\langle \begin{pmatrix}
    1 & 0 \\ 0 & 1
    \end{pmatrix}, 
    \begin{pmatrix}
    1 & 1 \\ 1 & 1
    \end{pmatrix} \right\rangle = 1\cdot 1 + 0 \cdot 1 + 0 \cdot 1 + 1 \cdot 1=0.
\end{align*}

In particular, 
the map $\mat \times \mat \to \F_q$ defined by $(X,Y) \mapsto \langle X,Y \rangle$ is bilinear, symmetric and non-degenerate.
\end{remark}

\begin{definition}
The \textbf{dual} of a code $\mC \le \mat$ is $\mC^{\perp}=\{Y \in \mat \mid \langle X,Y\rangle = 0 \textnormal{ for all } X \in \mC \}$.
Note that $\mC^\perp$ is a code as 
well.
\end{definition}

The following result is an immediate consequence of Remark~\ref{rem:nond}.

\begin{proposition}
For a code $\mC \le \mat$ we have $\dim(\mC^{\perp}) = nm-\dim(\mC)$.
\end{proposition}

\section{Rank Distribution and MacWilliams Identities}
\label{sec:2}

One of the best known results in coding theory, namely the \textit{MacWilliams identities}, establishes a relation between the (Hamming) weight distribution of a code and the weight distribution of its 
dual~\cite{macwilliams1963theorem}. This section is devoted to the rank-metric analogue of this result (in various formulations), first
discovered by Delsarte in~\cite{delsarte1978bilinear} using the theory of association schemes.
In this chapter, we follow the combinatorial approach proposed in
\cite{ravagnani2016rank,ravagnani2018duality}, which has been partially surveyed already in \cite{gorla2018codes}. The main new contribution of this section is showing the equivalence between the MacWilliams identities and their ``binomial moments'' version; see Remark~\ref{rem:equiv} below.
Note that we focus on $\F_q$-linear ``matrix'' codes. MacWilliams-type identities for the class of $\F_{q^m}$-linear ``vector'' rank-metric codes were established in~\cite{gadouleau2008macwilliams}.
The connection between the two MacWilliams-type identities is explained in~\cite[Remark~36]{ravagnani2016rank}.

\begin{definition} \label{def:gauss}
Let $a$ and $b$ be non-negative integers with $a \ge b$. The \textbf{$q$-ary coefficient} of $a$ and $b$
is the number of $b$-dimensional subspaces of any $a$-dimensional space over $\F_q$, denoted by
$$\qbin{a}{b}{q}= \prod_{i=0}^{b-1}\frac{q^a-q^i}{q^b-q^i},$$
where we put 
$$\qbin{a}{b}{q}=0 \quad \mbox{whenever $a<0$ or $b<0$}.$$
\end{definition}
The next lemma summarizes some well-known properties of the $q$-ary coefficient. A standard reference is \cite{stanley2011enumerative}.

\begin{lemma}
\label{lem:qb}
Let $a,b$ and $c$ be integers with $0 \le c\le b \le a$.
 \begin{enumerate}
\item We have $\qbin{a}{b}{q}=\qbin{a}{a-b}{q}$.
\item Let $U \le \F_q^a$ be a subspace of dimension $c$. The number of subspaces in $\F_q^a$ of dimension~$b$ that contain $U$ is $$\qbin{a-c}{b-c}{q}.$$ \label{qb1}
\item We have $\qbin{a}{b}{q}\qbin{b}{c}{q}=\qbin{a}{c}{q}\qbin{a-c}{a-b}{q}$. \label{qb2}
 \end{enumerate}
\end{lemma}

It might be the case that computing the rank distribution of a code is very hard, while it is easy for its dual. In such cases, the rank distribution of the code can be computed using the rank-metric analogue of the celebrated \textit{MacWilliams identities}, which we state next.

\begin{theorem}[MacWilliams-type identities]
\label{macwill}
Let $\mC \le \mat$ be a rank-metric code and let $0 \le i \le n$ be an integer. We have
\begin{align*}
    \Wrk_i(\mC^{\perp}) = \frac{1}{|\mC|}\sum_{j=0}^n \Wrk_j(\mC) \sum_{u=0}^i (-1)^{i-u}q^{mu+\binom{i-u}{2}}\qbin{n-j}{u}{q} \qbin{n-u}{i-u}{q}.
\end{align*}
\end{theorem}

The proof of Theorem~\ref{macwill} relies on the M\"obius inversion formula~\cite{stanley2011enumerative} and on the notion of \textit{shortened} rank-metric codes. The latter concept will be needed several times throughout this chapter.

\begin{definition}
Let $\mC \le \mat$ be a code and let $U \le \F_q^n$ be a subspace. The \textbf{$U$-shortening} of $\mC$ is the code $$\mC(U)=\{X \in \mC \mid \colsp(X) \le
     U\} \le \mC.$$
\end{definition}

The fact that $\mC(U)$ is a code directly follows from~\eqref{eq:col}.

\begin{notation}
In the sequel, we denote by $U^*$ the 
orthogonal of an $\F_q$-vector space 
$U \subseteq \F_q^n$ with respect to the standard inner product of 
$\F_q^n$. 
\end{notation}

\begin{lemma} \label{dimensioni}
 Let  $U \le \F_q^n$ be a subspace.  The following hold.
 \begin{enumerate}
\item $\dim
(\mat(U))=m \cdot \dim(U)$. \label{dimensioni1}
\item $\mat(U)^\perp = \mat(U^*)$. \label{dimensioni2}
 \end{enumerate}
\end{lemma}
\begin{proof}
\begin{enumerate}
 \item $\mat(U)$ is the space of matrices whose column space is contained in $U$. This is the same as saying that it is the space of matrices 
 whose columns belong to $U$. The latter set has cardinality $q^{m \cdot \dim(U)}$.
 
\item  The inclusion $\mat(U^*) \le \mat(U)^\perp$ follows from Remark~\ref{rem:nond}.
Equality then follows from a dimension argument, thanks to the previous part of the lemma. \qedhere
\end{enumerate}
\end{proof}

The following proposition can be seen as a rank-metric analogue of the duality between puncturing and shortening for codes in the Hamming metric. Even though it is quite simple to establish, several results of the duality theory of rank-metric codes are a consequence of this identity.

\begin{proposition} \label{prop:dual}
Let $\mC \le \mat$ be a code and let $U \le \F_q^n$. Set $u=\dim(U)$. We have 
$$|\mC(U)|= \frac{|\mC|}{q^{m(n-u)}}|\mC^{\perp}(U^*)|.$$
\end{proposition}
\begin{proof}
 We have  
 $$\mC(U)^\perp = (\mC \cap \mat(U))^\perp=\mathcal{\mC}^\perp
+\mat(U)^\perp=\mC^\perp + \mat(U^*),$$
where the last equality comes from Part~\ref{dimensioni2} of 
Lemma~\ref{dimensioni}. It follows
 \begin{equation} \label{eq2}
 q^{mn}=|\mC(U)| \cdot |\mC(U)^\perp|=|\mC(U)| \cdot |\mC^\perp +
\mat(U^*)|.
 \end{equation}
 On the other hand, Part~\ref{dimensioni1} of Lemma~\ref{dimensioni} gives
 $$\dim (\mC^\perp + \mat(U^*))=\dim (\mC^\perp)+m \cdot \dim
(U^*)-\dim  (\mC^\perp(U^*)).$$
As a consequence, 
  \begin{equation}\label{eq3}
|\mC^\perp + \mat(U^*)| = \frac{q^{nm} \cdot
q^{m(n-u)}}{|\mC| \cdot
|\mC^\perp(U^*)|}.
 \end{equation}
 Combining equations (\ref{eq2}) and (\ref{eq3}), one obtains the desired
result.
\end{proof}

\begin{proof}[Proof of Theorem \ref{macwill}]
    For all subspaces $V \le \F_q^n$ define
$$f(V)= |\{ M \in \mC^\perp \ | \ \mbox{colsp}(M) = V \}|, \ \ \ \ \ \ 
g(V)= \sum_{U \le V} f(U) = |\mC^\perp(V)|.$$
We now apply the M\"obius inversion formula in the lattice of subspaces of $\F_q^n$; see e.g.~\cite{stanley2011enumerative}. 
We obtain that for any $i \in \{0,...,n\}$ and for any $V \le \F_q^n$ 
of dimension $i$,  
\allowdisplaybreaks
\begin{eqnarray*}
f(V) &=& \sum_{u=0}^i {(-1)}^{i-u} q^{\binom{i-u}{2}}  \sum_{\substack{U
\le V \\ \dim(U)=u}} |\mC^\perp(U)| \\
&=& \sum_{u=0}^i {(-1)}^{i-u} q^{\binom{i-u}{2}}
\sum_{\substack{T
\le \F_q^n \\ T \ge V^* \\ \dim(T)=n-u}}
|\mC^\perp(T^*)| \\
&=& \frac{1}{|\mC|} \sum_{u=0}^i {(-1)}^{i-u} q^{mu+\binom{i-u}{2}}
\sum_{\substack{T
\le \F_q^n \\ T \ge V^* \\ \dim(T)=n-u}}
|\mC(T)|,
\end{eqnarray*}
where the last equality follows from Proposition 
\ref{prop:dual}. Next observe that
\allowdisplaybreaks
\begin{eqnarray} 
W_i(\mC^\perp) &=& \sum_{\substack{V
\le \F_q^n \\ \dim(V)=i}} f(V) \nonumber \\ &=& 
\frac{1}{|\mC|} \ \sum_{u=0}^i  {(-1)}^{i-u} q^{mu+\binom{i-u}{2}}  
  \sum_{\substack{V
\le \F_q^n \\ \dim(V)=i}} \sum_{\substack{T
\le \F_q^n \\ T \ge V^* \\ \dim(T)=n-u}}
|\mC(T)| \nonumber \\ &=&
\frac{1}{|\mC|} \ \sum_{u=0}^i  {(-1)}^{i-u} q^{mu+\binom{i-u}{2}}
\sum_{\substack{T
\le \F_q^n \\ \dim(T)=n-u}} 
\sum_{\substack{V
\le \F_q^n \\ V \ge T^* \\ \dim(V)=i}} |\mC(T)|  \nonumber  \\
&=& \frac{1}{|\mC|} \ \sum_{u=0}^i  {(-1)}^{i-u} q^{mu+\binom{i-u}{2}}
 \qbin{n-u}{i-u}{q}  \sum_{\substack{T
\le \F_q^n \\ \dim(T)=n-u}}  |\mC(T)|, \label{mw1}
\end{eqnarray}
where the last equality follows from Lemma~\ref{lem:qb} combined with the fact that $\dim(T^*)=u$.
Moreover,
\allowdisplaybreaks
\begin{eqnarray}
\sum_{\substack{T
\le \F_q^n \\ \dim(T)=n-u}}  |\mC(T)| &=& 
\sum_{\substack{T
\le \F_q^n \\ \dim(T)=n-u}} \sum_{j=0}^{n-u}
\sum_{\substack{S
\le T \\ \dim(S)=j}} |\{M \in \mC \ | \ \mbox{colsp}(M) = S\}| 
\nonumber \\
&=& \sum_{j=0}^{n-u} \sum_{\substack{S
\le \F_q^n \\ \dim(S)=j}} \sum_{\substack{T
\le \F_q^n \\ T \ge S \\ \dim(T)=n-u}} 
 |\{M \in \mC \ | \ \mbox{colsp}(M) = S\}|  \nonumber \\
 &=& \sum_{j=0}^{n-u} \qbin{n-j}{u}{q} W_j(\mC) = \sum_{j=0}^{n} \qbin{n-j}{u}{q} W_j(\mC). \label{mw2}
\end{eqnarray}
Combining equations (\ref{mw1}) and (\ref{mw2}), one obtains the desired result. \qedhere
\end{proof}

We will also need the following identity.

\begin{lemma}
\label{lem:cu}
Let $\mC \le \mat$ be a code and let $0 \le u \le n$ be an integer. Then
$$\sum_{\substack{U
\subseteq \F_q^n \\ \dim(U)=u}} |\mC(U)| = \sum_{i=0}^{n}  W_i(\mC)\qbin{n-i}{u-i}{q}$$
\end{lemma}

\begin{proof}
The result follows from counting the elements of
$$S=\{(U,X) \in \F_q^n \times \mC \mid \dim(U)=u,\ \mbox{colsp}(X) \subseteq U \}$$
in two different ways. The details are left to the reader.
\end{proof}

The following theorem can be interpreted as the implicit version of the MacWilliams-type identities and it is the main result of this section. 
\begin{theorem}
\label{thm:mcw}
Let $\mC \le \mat$ be a code and let $0 \le s \le n$ be an integer. We have
$$ \sum_{j=0}^{n-s}  W_j(\mC) \qbin{n-j}{s}{q} = \frac{|\mC|}{q^{ms}} \sum_{i=0}^{s}  W_i(\mC^\perp)\qbin{n-i}{s-i}{q}.$$
\end{theorem}
\begin{proof}
We can get the equivalent of the LHS of the desired equality by applying Lemma~\ref{lem:cu} where we set $u=n-s$. We get
$$\sum_{\substack{U \in \F_q^n
\\ \dim(U)=n-s }} |\mC(U)|= \sum_{j=0}^{n-s}  W_j(\mC) \qbin{n-j}{s}{q}.$$ By Proposition \ref{prop:dual} it remains to show that 
$$\frac{|\mC|}{q^{ms}} \sum_{i=0}^{s}  W_i(\mC^\perp)\qbin{n-i}{s-i}{q}
=\sum_{\substack{U \in \F_q^n
\\ \dim(U)=n-s }} \frac{|\mC|}{q^{ms}}\, |\mC^{\perp}(U^*)|  .$$
We have
\allowdisplaybreaks
\begin{eqnarray*}
\sum_{i=0}^{s}  W_i(\mC^\perp)\qbin{n-i}{s-i}{q} &=& \sum_{i=0}^{n}  W_i(\mC^\perp)\qbin{n-i}{s-i}{q} \\
&=&  \sum_{\substack{U \in \F_q^n
\\ \dim(U)=s }} |\mC^\perp(U)| \\
&=& \sum_{\substack{U \in \F_q^n
\\ \dim(U)=n-s }} |\mC^\perp(U^*)|,
\end{eqnarray*}
where the first equality follows from the definition of the $q$-ary coefficient and the last equality follows from the one-to-one correspondence between $s$-dimensional subspaces and the \smash{$(n-s)$-dimensional} subspaces of $\F_q^n.$ 
\end{proof}

\begin{remark} \label{rem:equiv}
It is known that Theorem \ref{thm:mcw} implies Theorem~\ref{macwill}; see
\cite[Appendix~A]{ravagnani2016rank}. We will conclude this section by showing that Theorem~\ref{thm:mcw} can be obtained from  the MacWilliams-type identities, i.e., the two statements are equivalent.
$$\Wrk_i(\mC^{\perp}) = \frac{1}{|\mC|}\sum_{j=0}^n \Wrk_j(\mC) \sum_{u=0}^i (-1)^{i-u}q^{mu+\binom{i-u}{2}}\qbin{n-j}{u}{q} \qbin{n-u}{i-u}{q}.$$
We start by multiplying both sides with $\qbin{n-i}{s-i}{q}$ and summing over $i$ for $0 \le i \le n$, obtaining
$$|\mC| \sum_{i=0}^n \Wrk_i(\mC^{\perp}) \qbin{n-i}{s-i}{q} = \sum_{j=0}^n \Wrk_j(\mC) \sum_{i=0}^n \sum_{u=0}^i (-1)^{i-u}q^{mu+\binom{i-u}{2}}\qbin{n-j}{u}{q} \qbin{n-u}{i-u}{q}\qbin{n-i}{s-i}{q}.$$
Next observe that $0 \le u \le i \le n.$ Thus the summation can be exchanged as follows:
$$|\mC| \sum_{i=0}^n \Wrk_i(\mC^{\perp}) \qbin{n-i}{s-i}{q} = \sum_{j=0}^n \Wrk_j(\mC) \sum_{u=0}^n q^{mu} \qbin{n-j}{u}{q} \sum_{i=u}^n (-1)^{i-u}q^{\binom{i-u}{2}} \qbin{n-u}{n-i}{q}\qbin{n-i}{s-i}{q}.$$
By substituting $\lambda = i-u$ and observing when the $q$-ary coefficients are zero, we have
{\small
$$|\mC| \sum_{i=0}^n \Wrk_i(\mC^{\perp}) \qbin{n-i}{s-i}{q} = \sum_{j=0}^n \Wrk_j(\mC) \sum_{u=0}^s q^{mu} \qbin{n-j}{u}{q} \sum_{\lambda=0}^{s-u} (-1)^{\lambda}q^{\binom{\lambda}{2}} \qbin{n-u}{n-\lambda-u}{q}\qbin{n-\lambda-u}{s-\lambda-u}{q}.$$}
By Part~\ref{qb2} of Lemma~\ref{lem:qb} and by reorganizing the terms, the previous identity can be rewritten as
$$|\mC| \sum_{i=0}^n \Wrk_i(\mC^{\perp}) \qbin{n-i}{s-i}{q} = \sum_{j=0}^n \Wrk_j(\mC) \sum_{u=0}^s q^{mu} \qbin{n-j}{u}{q} \qbin{n-u}{n-s}{q} \sum_{\lambda =0}^{s-u} (-1)^{\lambda}q^{\binom{\lambda}{2}} \qbin{s-u}{\lambda}{q}.$$
Recall that we want to obtain Theorem~\ref{thm:mcw}, i.e., we want to show that 

$$\sum_{j=0}^n \Wrk_j(\mC) \sum_{u=0}^s q^{mu} \qbin{n-j}{u}{q} \qbin{n-u}{n-s}{q} \sum_{\lambda =0}^{s-u} (-1)^{\lambda}q^{\binom{\lambda}{2}} \qbin{s-u}{\lambda}{q}= \sum_{j=0}^n \Wrk_j(\mC) q^{ms}\qbin{n-j}{s}{q}.$$
For this, recall the statement of the $q$-Binomial Theorem:
$$\sum_{\lambda =0}^{s-u} (-1)^{\lambda}q^{\binom{\lambda}{2}} \qbin{s-u}{\lambda}{q} = 
\begin{cases}
1 & \mbox{if $s=u$,} \\
0 & \mbox{otherwise.}
\end{cases}$$
By Lemma \ref{lem:qb}, we have
$$\sum_{u=0}^s q^{mu} \qbin{n-j}{u}{q} \qbin{n-u}{n-s}{q} = q^{ms}\qbin{n-j}{s}{q}.$$
Thus, Theorem \ref{macwill} and Theorem \ref{thm:mcw} are equivalent. 
\end{remark}

We conclude this section 
with an example that
illustrates how to apply
Theorem~\ref{thm:mcw}. 

\begin{example}
Let $n=m=3$. Consider the linear code $$\mC = \Bigg\langle \begin{pmatrix}
0 & 0 & 1\\ 2 & 0 & 0 \\ 0 & 0 & 0
\end{pmatrix}, \begin{pmatrix}
2 & 0 & 0 \\ 1 & 2 & 1 \\ 1 & 0 & 2
\end{pmatrix}
\Bigg\rangle \leq \F_3^{3 \times 3}.$$ The rank distribution of $\mC$ is $W_0(\mC)=1,\ W_1(\mC)=0,\ W_2(\mC)=4,\ W_3(\mC)=4.$ Since $\mC$ is of dimension 2, the dual code has dimension 7, and thus $|\mC^\perp|=3^7.$ This can also be seen by applying Theorem $\ref{thm:mcw}$ with $s=3$ to get $$ 1 = \frac{1}{3^7} (W_0(\mC^\perp)+W_1(\mC^\perp)+W_2(\mC^\perp)+W_3(\mC^\perp)).$$ Noting that $W_0(\mC^\perp)=1$ and applying Theorem \ref{thm:mcw} with $s=1$, we get $W_1(\mC^\perp)=38$. One more application of Theorem \ref{thm:mcw} with $s=2$ gives $W_2(\mC^\perp)=888$ to conclude that the rank distribution of $\mC^\perp$ is $(1,38,888,1260).$ We verify whether or not $W_3(\mC^\perp)=1260$ using the MacWilliams Identities. We have
\allowdisplaybreaks
\begin{eqnarray*}
W_3(\mC^{\perp}) &=&  \frac{1}{9}\sum_{j=0}^3 W_j(\mC) \sum_{u=0}^3 (-1)^{3-u}3^{3u+\binom{3-u}{2}}\qbin{3-j}{u}{3}
\\
&=& \frac{1}{9}[1(-27+ 1053 - 9477+ 19683)+4(-27+81)+4(-27)] \ = \ 1260.\\
\end{eqnarray*}
\end{example}

\section{Minimum Distance and MRD Codes}
\label{sec:3}

We devote this section to the minimum distance of rank-metric codes, which is the most relevant parameter for error correction in connection with various applications; see~\cite{silva2008rank,roth1991maximum,gabidulin1985theory}. We start with the Singleton-type bound for the cardinality of a rank-metric code of given minimum distance and with the construction of optimal codes. 
In the second part of this section we investigate the structural properties of MRD codes. In particular, we prove that the MRD property is closed under duality and that the rank distribution for MRD codes is completely determined by $(q,n,m)$, and by the minimum distance~$\drk(\mC)$. This section
follows to various degrees \cite{delsarte1978bilinear}, \cite{ravagnani2016rank} and \cite{ravagnani2018duality}.
For a more detailed treatment of MRD codes we recommend~\cite{sheekey201913} and the references therein.

\begin{theorem}[Singleton-type Bound]
\label{thm:slb}
Let $\mC \le \mat$ be a rank-metric code. We have
\begin{equation*} \label{singletonlikebound}
    \dim(\mC) \le {m(n-\drk(\mC)+1)}.
\end{equation*}
\end{theorem}
\begin{proof}
The statement of the theorem follows by definition if $\mC=\{0\}$. Suppose that $\mC \ne \{0\}$ and let $d=\drk(\mC)$. Consider the projection $\smash{\pi:\mat \longrightarrow \F_q^{(n-d+1)\times m}}$ onto the last $n-d+1$ rows. We claim that the restriction of $\pi$ to $\mC$ is injective, which would prove the bound. Indeed, suppose $\pi(X)=0$ for some $X \in \mC$. This implies that the last $n-d+1$ rows of $X$ are zero. In particular, $\rk(X)\le n-(n-d+1)=d-1$. Since $\drk(\mC)=d$, we have $X=0$.
\end{proof}

The Singleton-type bound says that a rank-metric code cannot have large minimum distance and large cardinality simultaneously. The best-known rank-metric codes are those meeting the Singleton-type bound, i.e., those codes having the maximum possible cardinality allowed by their minimum distance.

\begin{definition}
A code $\mC \le \mat$ 
is called \textbf{maximum rank distance}  (\textbf{MRD} in short)
if it attains the bound 
of Theorem~\ref{thm:slb} with equality.
\end{definition}

MRD codes are the most studied rank-metric codes and can be seen as the rank-metric analogues of MDS codes in the Hamming metric.

\begin{example} \label{ex:specs}
Consider the code 
\begin{align*}
    \mC = \left \langle
\begin{pmatrix}
1 & 0 & 0\\ 0 & 1 & 0 
\end{pmatrix},
\begin{pmatrix}
0 & 1 & 0\\ 0 & 0 & 1 
\end{pmatrix},
\begin{pmatrix}
0 & 0 & 1\\ 1 & 0 & 2 
\end{pmatrix}
\right \rangle \le \F_3^{2 \times 3}.
\end{align*}
Note that any matrix $X \in \mC$ can be written as
\begin{align*}
X=
\lambda_1 \begin{pmatrix}
1 & 0 & 0\\ 0 & 1 & 0 
\end{pmatrix} + \lambda_2
\begin{pmatrix}
0 & 1 & 0\\ 0 & 0 & 1 
\end{pmatrix} + \lambda_3
\begin{pmatrix}
0 & 0 & 1\\ 1 & 0 & 2 
\end{pmatrix} 
=
\begin{pmatrix}
\lambda_1 & \lambda_2 & \lambda_3\\ \lambda_3 & \lambda_1 & \lambda_2+2\lambda_3 
\end{pmatrix}, 
\end{align*}
where $\lambda_1, \lambda_2, \lambda_3 \in \F_3$. It is not hard to see that the matrix $X$ has rank smaller or equal to 1 if and only if 
\begin{align} \label{eq:spectfree}
\begin{pmatrix}
\lambda_3 \\ \lambda_1 \\ \lambda_2+2\lambda_3 
\end{pmatrix} =
\begin{pmatrix}
0 & 0 & 1\\ 1 & 0 & 0 \\ 0 & 1 & 2 
\end{pmatrix}
\begin{pmatrix}
\lambda_1 \\
\lambda_2 \\
\lambda_3
\end{pmatrix} = c \begin{pmatrix}
\lambda_1 \\
\lambda_2 \\
\lambda_3
\end{pmatrix}
\end{align}
for some $c \in \F_3$. Let \begin{align*}
    M= \begin{pmatrix}
0 & 0 & 1\\ 1 & 0 & 0 \\ 0 & 1 & 2 
\end{pmatrix}.
\end{align*}
Since the characteristic polynomial of the matrix $M$ is irreducible over $\F_3$, $M$ does not have any eigenvalues. This proves that there do not exist $\lambda_1, \lambda_2, \lambda_3 \in \F_3$ satisfying Equation~\eqref{eq:spectfree} unless $\lambda_1=\lambda_2=\lambda_3=0$. Therefore any non-zero $X \in \mC$ has rank 2 and the code $\mC$ is MRD.
\end{example}

MRD codes exist for all parameter sets, as the next result shows. This is in strong contrast to
MDS codes, which only exist over sufficiently large fields.
We refer to~\cite{ball2020arcs} for an expository paper on the problem of constructing MDS codes (or, equivalently, \textit{arcs} in projective spaces) over small fields.

 \begin{theorem} \label{thm:existence}
 Let $1 \le d \le n$. There exists an MRD code $\mC \le \mat$ of minimum distance $\drk(\mC) = d$.
 \end{theorem}

\begin{proof}
Fix any ordered basis $\{\beta_1,...,\beta_m\}$
of $\F_{q^m}$ over $\F_q$.
Let $\mbox{trace}: \F_{q^m} \to \F_q$ denote the trace map and let
$\{\beta_1^*,...,\beta_m^*\}$ be the 
\textit{dual} of the basis $\{\beta_1,...,\beta_m\}$;
see~\cite[p. 54]{lidl1997finite}.
Note that, by definition,
$\mbox{trace}(\beta_i\beta_j^*)=1$ if $i=j$ and $\mbox{trace}(\beta_i\beta_j^*)=0$ if $i \neq j$. Let $k=n-d+1$ and for a vector $\smash{u=(u_0,...,u_{k-1}) \in \F_{q^m}^{k}}$
define the matrix $X(u) \in \mat$ as
$$X(u)_{ij}:= \mbox{trace} \left( \beta_j  \sum_{\ell=0}^{k-1} u_\ell \, \beta_i^{q^\ell} \right) \quad \mbox{for $1 \le i \le n$ and $1 \le j \le m$}.$$
We claim that $\mC=\{X(u) \mid u \in \F_{q^m}^{k}\}$ has the desired properties, which is what we show in the rest of the proof.

It follows from the definitions that
$\mC$ is $\F_q$-linear. We claim that
$X(u)$ has rank at least~$d$ for all $u \in \F_{q^m}^{k}$, $u \neq 0$.
This implies at once that
$\drk(\mC) \ge d$ and that  $|\mC|=q^{mk}=q^{m(n-d+1)}$.
To prove the claim, for $u=(u_0,...,u_{k-1}) \in \F_{q^m}^k$ define the $\F_q$-linear map
$$L_u: \langle \beta_1,...,\beta_n \rangle_{\F_q} \to \F_{q^m}, \qquad L_u: x \mapsto \sum_{\ell=0}^{k-1} u_\ell \, x^{q^\ell}.$$
Then $X(u)$ is the matrix of $L_u$ with respect to the ordered bases $\{\beta_1,...,\beta_n\}$ and $\{\beta_1^*,...,\beta_m^*\}$, where the images are put in the columns.
It follows that $\rk(X(u))=n-\dim (\ker(L_u))$. Now observe that if $u \neq 0$, then the cardinality of $\ker(L_u)$
cannot exceed $q^{k-1}$, since
$\sum_{\ell=0}^{k-1} u_\ell \, x^{q^\ell}$ is a nonzero polynomial of degree at most $q^{k-1}$.
Therefore we have shown that $u \neq 0$ implies $\rk(X(u)) \ge n-(k-1)=n-(n-d+1-1)=d$.

It remains to show that $\drk(\mC)$ is exactly $d$ (and not larger). To see this 
we apply Theorem~\ref{thm:slb} and obtain 
\begin{align*}
    m(n-d+1) = \dim(\mC) \le m(n-\drk(\mC)+1) \le m(n-d+1),
\end{align*}
which implies $\drk(\mC)=d$.
\end{proof}

By Theorem~\ref{thm:existence} we know that MRD codes exist for any parameter set. 
The remainder of this section is devoted to structural properties of MRD codes. 

\begin{lemma} \label{lem:U}
Let $\mC \le \mat$ be a non-zero MRD code of minimum distance $d$. For all subspaces $U \le \F_q^n$ with $\dim(U)=u \ge d-1$ we have
\begin{align*}
    |\mC(U)| = q^{m(u-d+1)}.
\end{align*}
\end{lemma}
\begin{proof}
Define the space $V = \{(x_1,\ldots,x_u,0,\ldots,0)\in \F_q^n \mid x_1, \dots, x_u \in \F_q\}$. Let $g: \mathbb{F}_q^n \rightarrow \mathbb{F}_q^n$ be any \smash{$\mathbb{F}_q$-isomorphism} with $g(U) = V$ and denote by $G \in \mathbb{F}^{n \times n}_q$ the matrix associated to $g$ with respect to the canonical basis of $\mathbb{F}_q^n$. We define the rank-metric code $$\mD=G\mC = \{G X \mid X \in \mC\}.$$ 
It is easy to see that $\mathcal{D}$ has the same dimension and minimum distance as $\mathcal{C}$ which shows that it is MRD. Moreover, observe that $\mC(U) = \mD(V)$. Consider the maps
\begin{align*}
    \pi_1 : \mathcal{D} \rightarrow \F_q^{(n-d+1) \times m}, \quad \pi_2 : \F_q^{(n-d+1) \times m} \rightarrow  \F_q^{(n-u) \times m}
\end{align*}
where $\pi_1$ is the projection on the last $n-d+1$ rows, and $\pi_2$ is the projection on
the last $n-u$ rows. Since $\mD$ has minimum distance $d$, $\pi_1$ is injective. Furthermore, since $|\mD|=q^{m(n-d+1)}=|\F_q^{(n-d+1)\times m}|$, $\pi_1$ is bijective. The map $\pi_2$ is $\F_q$-linear and surjective.
Therefore
\begin{align*}
    |\pi_2^{-1}(0)| = |\pi_2^{-1}(M)| = q^{m(u-d+1)} \quad \textnormal{ for all } M \in \F_q^{(n-u) \times m}.
\end{align*}
Since $\pi_1$ is bijective and $\pi_2$ is surjective, the map $\pi=\pi_1 \circ \pi_2$ is surjective, and we have
\begin{align*}
    |\pi^{-1}(0)| = |\pi^{-1}(M)|=q^{m(u-d+1)} \quad \textnormal{ for all } M \in \F_q^{(n-u) \times m}.
\end{align*}
The statement follows from noting that $\mC(U)=\mD(V)=\pi^{-1}(0)$.
\end{proof}

Using the previous lemma we can show that the MRD property is closed under duality.

\begin{theorem} \label{dualofis}
Let $\mC \le \mat$ be an MRD code. Then $\mC^{\perp}$ is also MRD.
\end{theorem}
\begin{proof}
Clearly, if $\dim(\mC) \in \{0,mn\}$ the statement is true. Let us assume that $0 < \dim(\mC) < nm$ and let $d=\drk(\mC)$, $d^{\perp}=\drk(\mC^{\perp})$. The Singleton-type bound (see Theorem~\ref{thm:slb}) implies that
\begin{align*}
    \dim(\mC)=m(n-d+1) \textnormal{ and } \dim(\mC^{\perp}) \le m(n-d^{\perp}+1).
\end{align*}
In particular, $mn=\dim(\mC)+\dim(\mC^{\perp}) \le 2mn-m(d+d^{\perp})+2m$, which can be rewritten as
\begin{align} \label{eq:dual0}
    d+d^{\perp} \le n+2.
\end{align}
Let $U \le \F_q^n$ be of dimension $n-d+1$. By Proposition~\ref{prop:dual} we have
\begin{align*}
    |\mC^{\perp}(U)|=\frac{|\mC^{\perp}|}{q^{m(d-1)}}|\mC(U^{*})|.
\end{align*}
Since $\dim(U^{*})=d-1$, by Lemma~\ref{lem:U} we know $|\mC(U^{*})|=1$ and therefore
\begin{align} \label{eq:dual}
    |\mC^{\perp}(U)|=\frac{|\mC^{\perp}|}{q^{m(d-1)}}=\frac{q^{m(d-1)}}{q^{m(d-1)}}=1.
\end{align}
Since $U \le \mat$ with $\dim(U)=n-d+1$ was chosen arbitrarily, by~\eqref{eq:dual} we know that the only $M \in \mC$ with $\colsp(M) \subseteq U$ is the zero-matrix. Therefore $d^{\perp} \ge n-d+2$. In particular, by~\eqref{eq:dual0} we have $d^{\perp}=n-d+2$, from which it follows that 
\begin{align*}
    \dim(\mC^{\perp})= mn-\dim(\mC)=m(d-1)=m(n-d^{\perp}+1)
\end{align*}
proving that $\mC^{\perp}$ is MRD.
\end{proof}

The following result shows that the rank distribution of MRD codes in $\mat$ of minimum distance~$d$ is completely determined by the parameters $m$, $n$ and $d$.

\begin{theorem} \label{thm:macwillmrd}
Let $\mC \le \mat$ be a non-zero MRD code with minimum distance $d$. Then $\Wrk_0(\mC)=1$, $\Wrk_i(\mC)=0$ for $1 \le i \le d-1$ and 
\begin{align*}
    \Wrk_i(\mathcal{C}) = \sum_{u=0}^{d-1} (-1)^{i-u} q^{\binom{i-u}{2}} \qbin{n}{i}{q} \qbin{i}{u}{q} + \sum_{u=d}^i (-1)^{i-u} q^{\binom{i-u}{2}+m(u-d+1)} \qbin{n}{i}{q} \qbin{i}{u}{q}.
\end{align*}
\end{theorem}
\begin{proof}
Since the zero matrix is contained in every linear code, we have $\Wrk_0(\mC)=1$ and by the definition of minimum distance we have $\Wrk_i(\mC)=0$ for all $1 \le i \le d-1$. 
For all subspaces $V\le \F_q^n$ consider the maps
\begin{align*}
    f(V) &= |\{ X \in \mathcal{C} \mid \colsp(X) = V\}|, \\
    g(V) &= \sum_{U \subseteq V} f(U)=|\{X \in \mC \mid \colsp(X) \subseteq
     V\}|=|\mC(V)|.
\end{align*}
Let us fix $0 \leq i \leq n$ and an $i$-dimensional space $V \le \mathbb{F}_q^n$. By the Möbius inversion formula  in the lattice of subspaces we have
\begin{align} \label{eq:mcmrd}
    f(V) = \sum_{u=0}^i (-1)^{i-u} q^{\binom{i-u}{2}} \sum_{\substack{U \le V\\ \dim(U)=u}} g(U)
\end{align}
Since $\mC$ contains the zero matrix we have $g(U)=1$ for $0 \le \dim(U)\le d-1$ and since $\mC$ is MRD, by Lemma~\ref{lem:U} we have
\begin{align*}
    g(U) = q^{m(u-d+1)}
\end{align*}
if $d \le \dim(U) \leq n$.
Combining the above with~\eqref{eq:mcmrd} we obtain
\begin{align*}
    f(V) = \sum_{u=0}^{d-1} (-1)^{i-u} q^{\binom{i-u}{2}} \qbin{i}{u}{q} + \sum_{u=d}^{i} (-1)^{i-u} q^{\binom{i-u}{2}+m(u-d+1)} \qbin{i}{u}{q}
\end{align*}
By the identity 
\begin{align*}
\Wrk_i(\mC)=\sum_{\substack{V \le \F_q^n \\ \dim(V)=i}}f(V) = \qbin{n}{i}{q} f(V)
\end{align*}
the desired formula for $\Wrk_i(\mC)$ follows.
\end{proof}

It can be shown that, for an MRD code, all the elements of the rank distribution ``above'' the minimum distance are positive; see~\cite[Section~5]{ravagnani2016rank}.
We will need this result in Section~\ref{sec:4}.

\begin{theorem} \label{++}
Let $\mC \le \mat$ be an non-zero MRD code of minimum distance $d$. We have $\Wrk_i(\mC) > 0$ for all $d \le i \le n$.
\end{theorem}

\section{Maximum Rank and Anticodes}
\label{sec:4}

In this section we concentrate 
on the theory of 
anticodes in the rank metric, and in particular on the structure of the optimal ones. If the term ``code'' generally refers to a matrix space in which all elements have rank bounded from below by a given number, ``anticode'' refers to a matrix space in which all elements have rank bounded from above. 
The theory of anticodes (in various metrics) is closely connected to that of \textit{generalized weights} of error-correcting codes, which we however do not treat here (see~\cite{ravagnani2016generalized} for the connection). 
An overview of
the various notions of optimal anticodes and supports in the rank metric
can be found in~\cite{gorlach}.
A very recent development
in this area is the extension of the
theory of anticodes from the class of rank-metric codes to that of \textit{tensor codes}; see~\cite{byrne2021tensor}.

The main goal of this section is to survey Meshulam's characterization of optimal anticodes in the rank-metric, following the original proof presented in~\cite{meshulam1985maximal}. 
Note that the characterization was established outside the context of coding theory in 1985. 

In the second part of the section we
will then give two proofs of the fact that the dual of an optimal anticode is an optimal anticode~\cite{ravagnani2016rank}.
One proof uses Meshulam's characterization, the other one the connection between anticodes and MRD codes.

We start with a result that shows 
the spaces whose rank is bounded from above cannot have arbitrarily large dimension.

\begin{theorem} \label{anticodebound}
Let $\mC \le \mat$ be a code. We have $\dim(\mC) \le m \cdot \maxrk(\mC)$.
\end{theorem}
\begin{proof}
Let $\mu=\maxrk(\mC)$ and take an MRD code $\mD \le \mat$ such that $\drk(\mD) =\mu+1$ and $\dim(\mD)=m(n-\mu)$. Such a $\mD$ exists by Theorem~\ref{thm:existence}. We have $\mC \cap \mD =\{0\}$, so $\dim(\mC)+\dim(\mD) \le nm$. It follows that $\dim(\mC)\le nm-m(n-\mu)=m\mu$,
as desired.
\end{proof}

\begin{definition}
A code $\mC \le \mat$ that attains the bound of the previous theorem is called an \textbf{optimal anticode}.
\end{definition}

In this section, we have conveniently established Theorem~\ref{anticodebound} by using the existence of MRD codes. It is however interesting to present a different proof based on the notion of an \textit{initial set}. This concept will also allow us to characterize optimal anticodes, and to study the rank-metric covering radius of codes in Section \ref{sec:covering}.

\begin{definition} \label{def:initial}
The \textbf{initial entry} of a matrix $X \in \mat$ is $\operatorname{in}(X)=\min\{(i,j) \mid X_{ij} \ne 0\}$ where the minimum is taken with respect to the lexicographic order, i.e., 
\begin{align*}
    (1,1) < (1,2) < \dots < (1,m) < (2,1) < \dots < (n,m).
\end{align*}
The \textbf{initial set} of a non-zero code $\mC \le \mat$ is $\operatorname{in}(\mC) = \{\operatorname{in}(X) \mid X \in \mC, \; X \ne 0\}$.
\end{definition}



The following result is easy to see and the proof is left to the reader.

\begin{lemma}
For any non-zero code $\mC \le \mat$ we have $|\operatorname{in}(\mC)| = \dim(\mC)$.
\end{lemma}

We now state and prove a theorem of Meshulam, which gives a lower bound for the maximum rank of a code in terms of the combinatorics of its initial set. We start with the following concept.

\begin{notation}
The characteristic matrix of a subset $S \subseteq [a] \times [b]$ is $\chi(S)\in \F_2^{a \times b}$ defined by 
\begin{align*}
    \chi(S)_{ij}= \begin{cases} 1 &\textnormal{ if } (i,j) \in S \\
    0 &\textnormal{ otherwise.}
    \end{cases}
\end{align*}
We denote by $\lambda(S)$ the minimum number of lines (rows or columns) needed to cover all the~$1$'s appearing in $\chi(S)$.
\end{notation}

\begin{theorem}\label{mesh1}
Let $\mC \le \mat$ be a non-zero code. There exists a matrix $X \in \mC$ whose rank is at least $\lambda(\inn(\mC))$.
\end{theorem}

\begin{proof}
Define the map $\vect:\mat \to \F_q^{nm}$ that concatenates the rows of a matrix. Fix an ordered basis $\{B^1,...,B^k\}$ of $\mC$ and form the matrix having 
$\vect(B^1),...,\vect(B^k)$ as rows. Put the matrix in reduced row-echelon form, obtaining a new matrix with rows $g^1,...,g^k$. Let
$C^i=\vect^{-1}(g^i)$ for $i \in \{1,...,k\}$. Then 
$\{C^1,...,C^k\}$ is a basis of $\mC$ with the following properties:
\begin{itemize}
    \item $\inn(C^1) < \inn(C^2) < \cdots < \inn(C^k)$ in the lexicographic order;
    \item $\inn(\mC)=\{\inn(C^i) \mid 1 \le i \le k\}$;
    \item For all $1 \le i,j \le k$ with $i \neq j$, if $\inn(C^i)=(a,b)$ then $C^j_{a,b}=0$. 
\end{itemize}
By 
the famous
K\"onig-Egev\'ary Theorem,
$\lambda(\inn(\mC))$ is equal to the largest placement of non-attacking rooks on $\inn(\mC)$; see
\cite[Chapter~7]{riordan2014introduction}. Let this number be $r$ and fix such a placement $P \subseteq \inn(\mC)$.
Let $P=\{P_1,...,P_r\}$, with
$P_1 < P_2 < \cdots < P_r$ in the lexicographic order. For $i \in \{1,...,r\}$, let $p(i) \in \{1,...,k\}$ be the unique integer with
$\inn(C^{p(i)})=P_i$. We have $p(1) < p(2) < \cdots < p(r)$.
For all $1 \le i \le r$, let $D^i$ be the matrix obtained from $C^{p(i)}$ by removing 
any row or column that is not indexed by any element of~$P$. 
The $D^i$'s span a code $\mD \le \F_q^{r \times r}$ of dimension $r$ and it suffices to prove that $\mD$ contains an invertible matrix. By permuting the columns of all $D^i$'s simultaneously, we shall assume without loss of generality that $\inn(D^i)=(i,i)$. Note that the rows do not need to be permuted as $p(1) < p(2) < \cdots < p(r)$.
In particular, for all $1 \le i \le r$ the following holds: The first $i-1$ rows of $D^i$ are zero and $D^i_{i,j}=0$ for all $1 \le j <i$. 

Finally, denote by $Q$ the inverse of the square matrix whose $i$th row is the $i$th row of~$D^i$ for all $1 \le i \le r$ (the latter matrix is upper triangular with all 1's on the diagonal). Let $E^i=D^iQ$ for all $1 \le i \le r$. It suffices to show that $\langle E^1,...,E^r \rangle$
contains a matrix of rank $r$. For this note that
the following holds for all $i \in \{1,...,r\}$: The first $i-1$ rows of $E^i$ are zero and the $i$th row of~$E^i$ is the standard basis vector $e_i$. Therefore, the desired result follows from the following lemma.
\end{proof}

\begin{lemma}
Let $r \ge 1$ be an integer and let 
$E^1,...,E^r \in \F_q^{r \times r}$
be matrices with the following property: 
For all $i \in \{1,...,r\}$, the first $i-1$ rows of $E^i$ are zero and the $i$th row of~$E^i$ is the standard basis vector $e_i$. Then $\langle E^1,...,E^r \rangle$ contains an invertible matrix.
\end{lemma}

\begin{proof}
We proceed by induction on $r$. If $r=1$, then the result is trivial. If $r \ge 2$, then delete the last row and column from each of the matrices
$E^1,...,E^{r-1}$, obtaining matrices $F^1,...,F^{r-1}$. By the induction hypothesis,
$\langle F^1,...,F^{r-1} \rangle$ contains a matrix of rank $r-1$, say $X=\sum_{i=1}^{r-1}\lambda_iF^i$. 
Let $Y=\sum_{i=1}^{r-1}\lambda_iE^i$.
We compute the determinant of $E^r+Y$ by expanding along the bottom row:
$$\det(E^r+Y) =\det(X) + \det(Y).$$
Since $\det(X) \neq 0$, at least one among 
$E^r+Y$ and $Y$ is invertible. This concludes the proof.
\end{proof}

We are now ready to state and prove the main result of this section, which was  established in~\cite{meshulam1985maximal} (note that the terminology in~\cite{meshulam1985maximal} is different from the one we use, but the result is the same).

\begin{theorem} \label{mesh_thm}
Let $\mC \le \mat$ be an optimal anticode. Let $\mu=\maxrk(\mC)$.
\begin{itemize}
    \item[(i)] If $m > n$, then 
    $\mC=\{X \in \mat \mid \colsp(X) \subseteq U\}$
    for a subspace $U \le \F_q^n$ of dimension~$\mu$.
    \item[(ii)] If $m=n$, then either 
    $\mC=\{X \in \mat \mid \colsp(X) \subseteq U\}$ for
     a subspace $U \le \F_q^n$ of dimension $\mu$, or $\mC=\{X \in \mat \mid \rowsp(X)\subseteq U\}$ for a subspace $U \le \F_q^n = \F_q^m$ of dimension~$\mu$.
\end{itemize}
\end{theorem}

\begin{proof}
The result is clear if $\mu=n$ and we therefore assume $\mu <n$. By Theorem~\ref{mesh1}, we have $\lambda(\inn(\mC)) \le \mu$. Let $\{C^1,...,C^k\}$ be a basis of $\mC$ as in the proof of Theorem~\ref{mesh1}.
We have $k=|\inn(\mC)|=\mu m$ and thus $\inn(\mC)$ consists of $\mu$ horizontal lines (or, only in the case $m=n$, of $\mu$ vertical lines). In the sequel, we will assume that $\inn(\mC)$ consists of~$\mu$ horizontal lines, since the case of columns can be treated similarly.
By permuting the rows of the $C^i$'s (the same permutation for all matrices) and by performing Gaussian elimination again as in the proof of Theorem~\ref{mesh1}, we shall assume without loss of generality that:
\begin{itemize}
    \item For all $(a,b) \in [\mu] \times [m]$ we have $(a,b)=\inn(C^{am+b})$;
    \item $\inn(\mC)=\{\inn(C^i) \mid 1 \le i \le k\}$;
    \item For all $1 \le i,j \le k$ with $i \neq j$, if $\inn(C^i)=(a,b)$ then $C^j_{a,b}=0$. 
\end{itemize}
In the sequel, for $(a,b) \in [\mu] \times [m]$ we  let
$C^{(a,b)}=C^{am+b}$ and we denote by $M[i,j]$ the entry $(i,j)$ of a matrix $M$. We claim that  $C^{(a,b)}=0$ except possibly for the $b$th column. As this claim is invariant under row and column permutation, it suffices to prove that  
$C^{(\mu,\mu)}[\mu+1,\mu+1]=0$. For $(a,b) \in [\mu] \times [m]$, let
$B^{(a,b)}$ be the north-west $(\mu+1) \times (\mu+1)$ submatrix of
$C^{(a,b)}$. Furthermore, for $(i,j) \in [\mu] \times [\mu]$ denote by $E^{(i,j)}$ the $\mu \times \mu$ matrix having a 1 in position~$(i,j)$ and 0's elsewhere.

Define the set $P=\{(1,1), (2,2),..., (\mu-2,\mu-2), (\mu-1,\mu), (\mu,\mu-1)\}$ and 
$P'=P \cup \{(\mu,\mu)\}$.
In the case where $\mu=1$, we let $P=P'=\{(1,1)\}$. Note that for any $p \in [\mu] \times [\mu]$ and any $1 \le i \le \mu$ we have $B^{p}[i,\mu+1]=0$ and thus expanding along the $(\mu+1)$th column of $\sum_{p \in P} B^p$ yields
\begin{align*}
    0 &= \det\left( \sum_{p \in P} B^p \right) = \det \left( \sum_{p \in P} E^p \right) \cdot \left( \sum_{p \in P} B^p[\mu+1,\mu+1] \right), \\
    0 &= \det\left( \sum_{p \in P'} B^p \right) = \det \left( \sum_{p \in P'} E^p \right) \cdot \left( \sum_{p \in P'} B^p[\mu+1,\mu+1] \right),
\end{align*}
where the first equality above follows from the fact that $\mC$ does not contain a matrix of rank greater than $\mu$.
Since 
$$\det \left( \sum_{p \in P} E^p \right) \neq 0 \neq \det \left( \sum_{p \in P'} E^p \right),$$
it must be that $B^{(\mu,\mu)}[\mu+1,\mu+1]=0$, as desired. 

In order to finish the proof we show that for every $1 \le a \le \mu$, the $b$th column of $C^{(a,b)}$ is the same as the $\tilde{b}$th column of $C^{(a,\tilde{b})}$ for any $1 \le b,\tilde{b} \le m$. Since this claim is again invariant under row and column permutations, it is enough to show that $C^{(1,1)}[\mu+1,1] = C^{(1,2)}[\mu+1,2]$. Define the set $P=\{(1,1),(1,2),(2,3),(3,4), \dots, (\mu, \mu+1)\}$ and note that we have $C^{(1,1)}[\mu+1,1]=B^{(1,1)}[\mu+1,1]$ and $C^{(1,2)}[\mu+1,2]=B^{(1,2)}[\mu+1,2]$. Since $\mC$ does not contain a matrix of rank greater than 
$\mu$ we have 
\begin{align*}
        0 = \det\left(\sum_{p \in P}B^p\right) = (-1)^r\left(B^{(1,1)}[\mu+1,1]-B^{(1,2)}[\mu+1,2]\right),
\end{align*}
where the latter equality is obtained by expanding along the $(\mu+1)$th row of $\sum_{p \in P}B^p$.
It follows that $B^{(1,1)}[\mu+1,1]=B^{(1,2)}[\mu+1,2]$. Now for every $1 \le a \le \mu$, we have
\begin{align*}
    \colsp\left(C^{(a,b)}\right) = \colsp\left(C^{(a,\tilde{b})}\right) \quad \textnormal{for all $1 \le b,\tilde{b} \le m$.}
\end{align*}
Therefore if we set $$U=\colsp\left(C^{(1,i_1)}\right)  \oplus \colsp\left(C^{(2,i_2)}\right)\oplus \dots \oplus \colsp\left(C^{(\mu,i_{\mu})}\right),$$ where $i_{\ell} \in \{1, \dots, m\}$ for all $\ell \in \{1,\dots, \mu\}$ are arbitrary, then clearly $\dim(U)=\mu$ and we have $\mC=\{X \in \mat \mid \colsp(X) \subseteq U\}$.
\end{proof}

The classification of optimal anticodes established in the previous result allows us to show that the dual of an optimal anticode is an optimal anticode. This result was first proved via a different argument~\cite{ravagnani2016rank}, which we survey below.

\begin{proposition} \label{doa}
The dual of an optimal anticode is an optimal anticode.
\end{proposition}
\begin{proof}
If $\mC=\{X \in \mat \mid \colsp(X) \subseteq U\}$ for some $U \le \F_q^n$, then $\mC^{\perp}=\{X \in \mat \mid \colsp(X) \subseteq U^{*}\}$ where $U^*$ is the orthogonal of $U$ with respect to standard inner product of $\F_q^n$. Note that in the previous equality, one inclusion is easy to show and the other follows from dimension arguments. Therefore $\mC^{\perp}$ is an optimal anticode again. The same proof applies to show that $$\{X \in \mat \mid \rowsp(X) \subseteq U\}^{\perp} = \{X \in \mat \mid \rowsp(X) \subseteq U^*\},$$ if $n=m$. This concludes the proof.
\end{proof}

Interestingly, there is another proof of Proposition~\ref{doa} that does not use any result
about the structure of optimal anticodes. The key ingredients in the argument are the existence of MRD codes and Theorem~\ref{++}.

\begin{proof}[Different proof of Proposition~\ref{doa}]
Suppose that $\mC \le \mat$ is an optimal anticode having
$\maxrk(\mC)=\mu$. We will show that $\mC^\perp$ is an optimal anticode with $\maxrk(\mC^\perp)=n-\mu$. Towards a contradiction, suppose that $\mC^\perp$ contains a matrix $R$ of rank $r \ge n-\mu+1>0$. Let $\mD \le \mat$ be an MRD code of minimum distance $n-\mu+1$. It exists by
Theorem~\ref{thm:existence}. Moreover, by the 
very definition of MRD code we have
$\dim(\mD)=m\mu$. By Theorem~\ref{++}, the code $\mD$ contains a matrix $Y$ of rank $r$. There exist invertible matrices
$A$ and $B$ of sizes $n \times n$ and $m \times m$, respectively, such that $AYB=R$. Let
$\mD_1:=A \mD B=\{AXB \mid X \in \mD\}$. Then $\mD_1$ is an MRD code with the same parameters as $\mD$ and such that
$R \in \mD_1$. We then have $\dim(\mC^\perp \cap \mD_1) \ge 1$, and therefore $\dim(\mC+\mD_1^\perp) \le mn-1$.
The latter implies $\dim(\mC \cap \mD_1^\perp) \ge -mn+1+\dim(\mC)+\dim(\mD_1^\perp)=1$. By Theorem~\ref{dualofis}, 
$\mD_1^\perp$ is an MRD code of minimum distance $\mu+1$. This contradicts the fact that
$\maxrk(\mC)=\mu$, concluding the proof.
\end{proof}

\section{Covering Radius and External Distance}
\label{sec:covering}
\label{sec:5}

In this section, we focus on the covering radius of rank-metric codes. The goal is to prove three upper bounds on the covering radius, namely the \emph{dual distance bound}, the \emph{external distance bound} and the \emph{initial set bound}. The first two bounds have an analogue in the Hamming metric, whereas the external distance bound is somewhat specific to the rank-metric setup. The treatment of this topic is mainly inspired by~\cite{byrne2017covering}.
Another foundational reference for the packing and covering properties of rank-metric codes is~\cite{gadouleau2008packing}.

We start by showing that the covering radius is related to a parameter of the dual code in a rather surprising way. Recall that the covering radius of a code $\mC \le \mat$ is the parameter $$\rho(\mC) = \min\{i \mid \textnormal{ for all } X \in \mat \; \exists \, Y  \in \mC  \textnormal{ with } \drk(X,Y) \le i\}.$$ 
A more intuitive definition of the covering radius is given as follows: For a code $\mC \le \mat$ consider the union of balls of fixed radius centered at each codeword. The smallest value the radius can take provided that the union covers the space $\mat$ is the covering radius of the code $\mC$. This explains where the name ``covering radius" comes from.

\begin{remark}
\begin{itemize}
\item[(i)] Since $n \leq m$, a trivial upper bound for the covering radius of a code $\mC \le \mat$ is $\rho(\mC) \leq n$.
\item[(ii)] If $\mC$ is the whole ambient space, then the covering radius of $\mC$ is 0.
\end{itemize}
\end{remark}

Suppose that $\mC \subset \mC'$. Then by definition we have $\rho(\mC') < \rho(\mC).$ Moreover, if $X' \in \mC' \setminus \mC$ then there exists $X \in \mC$ such that $\drk(X,X') \le \rho(\mC).$ Since $\mC$ is contained in $\mC'$, $X$ is also contained in $\mC'.$ Thus, we have $\drk(\mC') \le \rho(\mC)$. 

It is natural to ask about the connection between the covering radius and minimum distance of a rank-metric code.

\begin{proposition}
Let $\mC \le \mat$ be a code. We have 
$$\rho(\mC) \ge \left\lceil \frac{d(\mC)-1}{2} \right\rceil.$$
\end{proposition}
\begin{proof}
It is not difficult to see that
$$\rho(\mC) \geq \left\lfloor \frac{\drk(\mC)-1}{2} \right\rfloor.$$ Therefore in order to prove the proposition it suffices to show that equality cannot hold. This follows from the non-existence of perfect rank-metric codes; see~\cite{loidreau2008}. 
\end{proof}

The next definition will be needed later to derive an upper bound on the covering radius. 

\begin{definition}
Given a rank-metric code $\mC \le \mat$ and a matrix $M \in \mat$, the \textbf{translate} of $\mC$ by $M$ is defined as the code $\mC + M =\{ A + M \mid A \in \mC \}.$
\end{definition}
It is very natural to study the rank distribution of the translates of $\mC.$ Note that if we know the rank distributions of all the translates of $\mC$, then we can form a set $S$ which consists of the minimum rank (weight) of each translate of the code $\mC$. Since the covering radius of $\mC$ is just the maximum distance of $\mC$ to any matrix in $\mat$, we would then have  $\max(S)=\rho(\mC).$ 
We start by giving a formula for the size of the $U$-shortening of any translate of a code $\mC \le \mat$.

\begin{lemma}\label{lem:trans1}
Let $\mC \le \mat$ be a code and let $U \le \F_q^n$ be of dimension $u$. If $|\mC(U)|=|\mC|/q^{m(n-u)}$, then we have 
$$|(\mC+M)(U)|= \frac{|\mC|}{q^{m(n-u)}}$$
for all $M \in \mat$.
\end{lemma}
\begin{proof}
Let $\varphi: \F_q^n \longrightarrow \F_q^n$ be any isomorphism such that
$$\varphi(U)=\{(x_1,\ldots,x_u,0,\ldots,0)\in \F_q^n \mid x_1, \dots, x_u \in \F_q\}$$ and let $V=\varphi(U)$. Let $A$ be the matrix representation of $\varphi$ with respect to the canonical basis of $\F_q^n$ and let $\mC' = A\mC.$ Recall that 
\begin{align*}
\mC(U)&=\{X \in \mC \mid \colsp(X) \subseteq
     U\}, \\
\mC'(V)&=\{Y \in A\mC \mid \colsp(Y) \subseteq \varphi(U)\}.
\end{align*}
Thus, multiplication by $A$ from the left leads to a bijection between these sets. In particular, we have $|\mC(U)|=|\mC'(V)|$ and $|(\mC+M)(U)| = |(\mC'+AM)(V)|$. Instead of showing that $|\mC(U)|=|(\mC+M)(U)|$, we prove the desired statement by showing  $$|\mC'(V)|=|(\mC'+AM)(V)|$$ for all $M \in \mat$. Define $\pi=\mat \rightarrow \F_q^{(n-u)\times m}$ to be the projection onto the last $n-u$ rows. Consider the following two maps: $$\pi_1 = \pi|_{\mC'}, \quad \pi_2 = \pi|_{\mC'+AM}.$$ Since $V=\varphi(U)$, we have $\ker(\pi_1)=\mC'(V)$ and $\pi_2^{-1}(0)=(\mC'+AM)(V).$ Observe that $|\pi_1^{-1}(-\pi(AM))|=|\pi_2^{-1}(0)|$ since 
\begin{align*}
\pi_1^{-1}(-\pi(AM))&=\{Y \in \mC' \mid \pi_1(Y) = -\pi(AM)\} \\ \pi_2^{-1}(0)&=\{Y \in \mC' \mid \pi_2(Y+AM)=0\}.
\end{align*}
Since $\ker(\pi_1)=\mC'(V)$, we have $|\pi_1(\mC')|=q^{(n-u) m}=|\F_q^{(n-u)\times m}|$. Therefore we conclude that $\pi_1$ is surjective and there exists $C \in \mC$ such that $\pi_1(C)=-\pi(AM)$. If $\pi_1(X)=0$ then $\pi_1(X+C)=-\pi(AM)$ meaning that $|\pi_1^{-1}(-\pi(AM))|=|\ker(\pi_1)|$. Hence we have 
$$|(\mC'+AM)(V)|=|\pi_2^{-1}(0)|=|\pi_1^{-1}(-\pi(AM))|=|\ker(\pi_1)|=|\mC'(V)|,$$ as desired.  \qedhere
\end{proof}

\begin{lemma}\label{lem:trans2}
Let $\mC \le \mat$ be a code. For any $U \le \F_q^n$ with $u=\dim(U) \ge n-\drk(\mC^{\perp})+1$ and for all $M \in \mat$, we have 
$$|(\mC+M)(U)|= \frac{|\mC|}{q^{m(n-u)}}.$$
\end{lemma}
\begin{proof}
It is enough to prove the desired equality for $M=0$ since then Lemma \ref{lem:trans1} implies the result. By Proposition \ref{prop:dual} we have $$|\mC(U)|= \frac{|\mC|}{q^{m(n-u)}}|\mC^{\perp}(U^*)|.$$ 
Note that $$\dim(U^*)= n - \dim(U)= n- u \le n -n + \drk(\mC^{\perp}) -1= \drk(\mC^{\perp}) -1.$$ It follows that $\mC^{\perp}(U^*)=\{0\}$, and thus $|\mC^{\perp}(U^*)|=1$.
\end{proof}

We now have established the needed machinery to show that the rank distribution of the translate $\mC + M$ can be expressed in terms of $\Wrk_{0}(\mC+M), \Wrk_{1}(\mC+M), \ldots,\Wrk_{n-\drk(\mC^{\perp})}(\mC+M)$. The following theorem gives the desired relation. For the remainder of this section, we use $d^\perp = \drk(\mC^{\perp})$ for the dual of a code $\mC$.

\begin{theorem} \label{thm:trans}
Let $\mC \le \mat$ be a code and $M \in \mat$. Then we have
{\small
\begin{equation*}
    \Wrk_i(\mC + M) = \sum\limits_{k=0}^{n-d^\perp}(-1)^{i-k}q^{\binom{i-k}{2}}\qbin{n-k}{i-k}{q}\sum\limits_{j=0}^{k}\Wrk_j(\mC+M)\qbin{n-j}{k-j}{q} + \sum\limits_{k=n-d^\perp+1}^{i}\qbin{n}{k}{q}\frac{|\mC|}{q^{m(n-k)}}
\end{equation*}}
for all $k-\drk(\mC^{\perp})+1 \le i \le n$.
\end{theorem}
\begin{proof}
Let $A \le B \le \F_q^n$ such that $\dim(A)=a$ and $\dim(B)=b$. The Möbius function in the lattice of subspaces gives $\mu(A,B) = (-1)^{b-a}q^{\binom{b-a}{2}}$. For any $U \le \F_q^n$, let 
\begin{align*}
f(U) &= |\{X \in \mC +M \mid \colsp(X)=U\}| \\ g(U) &= \sum_{V \subseteq U} f(V) = |(\mC +M)(V) |.
\end{align*}
By the Möbius inversion formula we have 
\begin{align*}
    f(U) = \sum_{V \subseteq U} g(V) \, \mu(V,U)
    = \sum_{V \subseteq U} |(\mC +M)(V)|\, \mu(V,U). 
\end{align*}
Fix an integer $i$ with $k-\drk(\mC^{\perp})+1 \le i \le n$. We have
\begin{eqnarray} 
W_i(\mC + M) &=& \sum_{\substack{U
\subseteq \F_q^n \\ \dim(U)=i}} f(U) \nonumber \\ &=& 
\sum_{\substack{U
\subseteq \F_q^n \\ \dim(U)=i}} \sum_{V \subseteq U}
|(\mC+M)(V)|\mu(V,U) \nonumber \\ &=&
\sum_{V \subseteq \F_q^n} |(\mC+M)(V)| \sum_{\substack{V
\subseteq U \\ \dim(U)=i}} \mu(V,U) \nonumber \\ &=&
\sum_{k=0}^i \sum_{\substack{V
\subseteq \F_q^n \\ \dim(V)=k}} |(\mC+M)(V)| \sum_{\substack{V
\subseteq U \\ \dim(U)=i}}(-1)^{i-k}q^{\binom{i-k}{2}} \nonumber \\&=&  
\sum_{k=0}^i (-1)^{i-k}q^{\binom{i-k}{2}}\qbin{n-k}{i-k}{q} \sum_{\substack{V
\subseteq \F_q^n \\ \dim(V)=k}} |(\mC+M)(V)|. \label{ab1}
\end{eqnarray}
In order to conclude the proof, we need to distinguish between two cases. If $k \le n-\drk(\mC^{\perp})$, then
\begin{eqnarray}
\sum_{\substack{V
\subseteq \F_q^n \\ \dim(V)=k}} |(\mC+M)(V)| &=& \sum_{X \in \mC +M} |\{V \in \F_q^n \mid \dim(V)=k \textnormal{ and } \colsp(X) \subseteq V\}| \nonumber \\ &=& \sum_{j=0}^k \sum_{\substack{\rk(X)=j
\\ X \in \mC + M}}  |\{V \in \F_q^n \mid \dim(V)=k \textnormal{ and } \colsp(X) \subseteq V\}| \nonumber \\ &=&
\sum_{j=0}^k \Wrk_j(\mC + M)\qbin{n-j}{k-j}{q}. \label{ab2}
\end{eqnarray}
On the other hand, if $k \ge n-\drk(\mC^{\perp})+1$, then by Lemma \ref{lem:trans2}, we have 
\begin{eqnarray}
\sum_{\substack{V
\subseteq \F_q^n \\ \dim(V)=k}} |(\mC+M)(V)| = \qbin{n}{k}{q} \frac{|\mC|}{q^{m(n-k)}}. \label{ab3}
\end{eqnarray}
Combining~\eqref{ab1}, \eqref{ab2} and \eqref{ab3} give us the desired formula.
\end{proof}

We have the following upper bound on the covering radius.
\begin{corollary}[Dual Distance Bound] \label{cor:dualdisbound}
For a code $\mC \le \mat$ with $\mC \ne \mat$ we have $$\rho(\mC) \le n-\drk(\mC^{\perp})+1.$$
\end{corollary}
\begin{proof}
Choose $M \in \mat \setminus \mC$. Since $\mC$ is linear, we have $\Wrk_{0}(\mC + M)=0.$ By applying Theorem~\ref{thm:trans} with $i= n - \drk(\mC^{\perp}) +1$ we obtain
\begin{multline*}
    \Wrk_{n - d^\perp +1}(\mC + M) \\ = \sum_{k=1}^{n-d^\perp}(-1)^{i-k}q^{\binom{i-k}{2}}\qbin{n-k}{i-k}{q}\sum_{j=1}^{k}\Wrk_j(\mC+M)\qbin{n-j}{k-j}{q} + \qbin{n}{n-d^\perp+1}{q}\frac{|\mC|}{q^{m(d^\perp-1)}}.
\end{multline*}
Since $\Wrk_{1}(\mC+M),\Wrk_{2}(\mC+M),\ldots,\Wrk_{n - \drk(\mC^{\perp}) +1}(\mC+M)$ cannot all be equal to zero at the same time, we obtain the desired bound.
\end{proof}

We can also get an upper bound for the covering radius of $\mC$ using the notion of \emph{external distance}, which is defined as follows.

\begin{definition}
The \textbf{external distance} of a code $\mC \le \mat$ is $$s(\mC) = |\{1 \le i \le n \mid \Wrk_i(\mC^{\perp})>0\}|.$$
\end{definition}

The following relation between the external distance and the covering radius of a code was given in~\cite[Theorem 5.6]{byrne2017covering}.
We omit the proof here.

\begin{theorem}[External Distance Bound] \label{thm:extdistbound}
Let $\mC \le \mat$ be a code. We have $\rho(\mC) \le s(\mC).$
\end{theorem}
Since for any code $\mC \le \mat$ we have $s(\mC) \leq n - d^\perp +1$, the external distance bound stated in Theorem~\ref{thm:extdistbound} improves on the dual distance bound stated in Corollary~\ref{cor:dualdisbound}.

These results are analogous to results that exist for codes in the Hamming metric. We conclude with a bound that uses the matrix structure of rank-metric codes. The initial entry of a matrix and the initial set of a code $\mC$ are defined in Definition \ref{def:initial}. We give an example to get the reader familiar with the concept. 
\begin{example} \label{ex:bounds} 
Consider the following matrix. $$A =\begin{pmatrix}
0 & 0 & 0 & 1\\ 0 & 1 & 1 & 0\\ 1 & 0 & 0 & 0
\end{pmatrix} \in \F_2^{3 \times 4}.$$ Then $\operatorname{in}(A)=(1,4)$. Let
$$\mC = \Big\langle \begin{pmatrix}
1 & 0 & 1\\ 0 & 1 & 1
\end{pmatrix}, \begin{pmatrix}
1 & 1 & 1 \\ 1 & 0 & 1
\end{pmatrix}, \begin{pmatrix}
0 & 1 & 1 \\ 1 & 0 & 0
\end{pmatrix} \Big\rangle \leq \F_2^{2 \times 3}.$$ It can easily be checked that $\drk(\mC)=2$ and $\operatorname{in}(\mC)=\{(1,1),(1,2),(1,3)\} \subseteq [1] \times [3]$. As we noted before, see that $\dim(\mC)=|\operatorname{in}(\mC)|$. Lastly, $\lambda(\operatorname{in}(\mC))=1$ since we only need the first row to cover the 1's. 
\end{example}

\begin{lemma} Let $\mC \le \mat$ be a non-zero code of minimum distance $d$. We have $$\operatorname{in}(\mC) \subseteq [n-d+1] \times [m].$$
\end{lemma}
\begin{proof}
We know that $\dim(\mC)=|\operatorname{in}(\mC)|$. Let $k=\dim(\mC)$. Without loss of generality, we can choose a basis $\{X_1,\ldots,X_k\}$ for $\mC$ such that $$(1,1) \le \operatorname{in}(X_1)< \operatorname{in}(X_2) < \cdots < \operatorname{in}(X_k).$$ 
We also know that $\operatorname{in}(\mC)=\{\operatorname{in}(X_1),\ldots,\operatorname{in}(X_k)\}.$ We prove the result by contradiction. Assume that $(n-d+1,m) < \operatorname{in}(X_k).$ This implies that the first $n-d+1$ rows of $X_k$ are 0-rows since any codeword has $m$ columns. Thus, $\rk(X_k) \le d-1$, contradicting the minimal distance of $\mC$. So, $\operatorname{in}(X_k) \le (n-d+1,m)$ as desired. 
\end{proof}
Our main goal is to find an upper bound for the covering radius of a rank-metric code using some properties of its initial set. We need the following lemma.
\begin{lemma} \label{lem:isb}
Let $a$ and $b$ positive integers and let $S \subseteq [a] \times [b]$. If $X\in \F_q^{a \times b}$ is such that $X_{ij} = 0$ whenever $(i,j) \not\in S$, then $\rk(X) \le \lambda(S).$ 
\end{lemma}
\begin{proof}
We prove the statement by induction on $\lambda(S).$ If $\lambda(S)=0$, then $M=0$, i.e., $\rk(M)=\lambda(S).$ Suppose the statement is true for a given value of $\lambda(S)$, say $r$. Consider the characteristic matrix~$\chi$ of $S$ such that $\lambda(S)=r+1$. Without loss of generality, assume that the first row of $\chi$ contributes to~$\lambda(S).$ Now, consider $\chi'$ to be the matrix obtained from $\chi$ by setting the first row to zero. Let~$S'$ be the set whose characteristic matrix is $\chi'.$ By the induction hypothesis we have $\lambda(S')=r$ and there exists $M'$ such that $\rk(M')\le r.$ The first row of $M'$ is a 0-row. Suppose we change the first row of $M'$ to any row provided that the assumption still holds, and call it~$M$. Then it holds that $\rk(M) \le r+1=\lambda(S)$ as desired.
\end{proof}

We are now ready to give a third upper bound on the covering radius of a code.

\begin{theorem}[Initial Set Bound] \label{thm:isb}
Let $\mC \le \mat$ be a non-zero code with minimum distance~$d$. We have $$\rho(\mC) \le d-1+\lambda(S)$$ where $S=([n-d+1] \times [m])\backslash \operatorname{in}(\mC).$
\end{theorem}
\begin{proof}
Choose an arbitrary $M \in \mat$. There exists $N \in \mC$ such that $M$ and $N$ coincide in the entries that belong to $\operatorname{in}(\mC)$, i.e., we have $(M-N)_{ij}=0$ for $(i,j) \in \operatorname{in}(\mC).$ Let $X$ be the matrix obtained from $M-N$ by deleting the first $d-1$ rows. Clearly we have $\rk(X)+(d-1) \geq \rk(M-N)$, i.e., $$\drk(M,N) \leq d-1+\rk(X).$$ 
Note that $X$ is an  $(n-d+1)\times m$-matrix over $\F_q$ satisfying $X_{ij}=0$ for all $(i,j) \in \operatorname{in}(\mC).$ In particular, $X$ satisfies the assumptions of Lemma \ref{lem:isb}, from which it follows that $\drk(M,N)\le d-1+\lambda(S)$. Since we chose $M \in \mat$ arbitrarily, this proves the theorem.
\end{proof}

\begin{example}
\label{exrem}
Consider the code $\mC \le \F_2^{2 \times 3}$ given in Example \ref{ex:bounds}. We have $$\mC^\perp = \Big\langle \begin{pmatrix}
1 & 0 & 0\\ 0 & 0 & 1
\end{pmatrix}, \begin{pmatrix}
0 & 1 & 0 \\ 1 & 0 & 0
\end{pmatrix}, \begin{pmatrix}
0 & 0 & 1 \\ 1 & 1 & 0
\end{pmatrix} \Big\rangle \leq \F_2^{2 \times 3}.$$
All nonzero matrices in $\mC^\perp$ have rank 2, and thus $s(\mC)=1.$ Therefore the external distance bound gives $\rho(C) \le 1$. Similarly, the initial set bound gives $\rho(C) \leq 2 - 1 + 0= 1$.
Moreover, the dual distance bound also gives $\rho(C) \le 1.$ This is an example where all 3 bounds give the same value. 
\end{example}

We already know that external distance bound improves on the dual distance bound. As shown in the following example, it might be the case that the initial set bound is better than the external distance bound, and thus naturally better than the dual distance bound. 
\begin{example}
\label{coverbound}
Consider $\mC \le \F_2^{3 \times 3}$ to be defined as follows.
$$\mC = \Big\langle \begin{pmatrix}
1 & 0 & 0\\ 0 & 1 & 0 \\ 0 & 0 & 0
\end{pmatrix}, \begin{pmatrix}
0 & 1 & 0\\ 0 & 0 & 1 \\ 0 & 0 & 0
\end{pmatrix}, \begin{pmatrix}
0 & 0 & 1\\ 0 & 0 & 0 \\ 1 & 0 & 0
\end{pmatrix} \Big\rangle.$$

We obtain the following dual of $\mC$.
{\small
\begin{equation*}
\mC^\perp = \Big\langle \begin{pmatrix}
1 & 0 & 0\\ 0 & 1 & 0 \\ 0 & 0 & 0
\end{pmatrix}, \begin{pmatrix}
0 & 1 & 0\\ 0 & 0 & 1 \\ 0 & 0 & 0
\end{pmatrix},\begin{pmatrix}
0 & 0 & 1\\ 0 & 0 & 0 \\ 1 & 0 & 0
\end{pmatrix},\begin{pmatrix}
0 & 0 & 0\\ 1 & 0 & 0 \\ 0 & 0 & 0
\end{pmatrix},\begin{pmatrix}
0 & 0 & 0\\ 0 & 0 & 0 \\ 0 & 1 & 0
\end{pmatrix},\begin{pmatrix}
0 & 0 & 0\\ 0 & 0 & 0 \\ 0 & 0 & 1
\end{pmatrix} \Big\rangle.
\end{equation*}}

It is an easy exercise to check that $\drk(\mC)=2$ and  $s(\mC)=3$. Moreover, we have that $\operatorname{in}(\mC)=\{(1,1),(1,2),(1,3)\}$. In particular, following the notation of
Theorem~\ref{thm:isb}, we have
$S= \{(2,1),(2,2),(2,3)\}$ and so
$\lambda(S)=1$. 
The initial set bound gives $\rho(\mC) \le 2$, whereas the external distance bound gives $\rho(\mC) \le 3.$ The dual distance bound also reads $\rho(\mC) \le 3$.
\end{example}

\section{Field Size and the Density of MRD Codes}
\label{sec:6}

It is well known that MDS codes in $\F_q^n$ are dense within the family of $\F_q$-linear block codes with the same dimension as $q \to +\infty$.
This means that for a large enough field size $q$, a uniformly random chosen block code of a certain dimension is MDS with high probability. In the rank metric, MRD codes can be regarded as the rank-metric analogues of MDS codes in the Hamming metric, with which they share many common properties. In this section, we will investigate the proportion of $\F_q$-linear MRD codes in $\mat$ within the set of codes of the same dimension and see what happens both when $q \to +\infty$ and $m \to +\infty$. The goal of this section is to show 
that MRD codes exhibit a very  different behavior to MDS codes from a density perspective. 
Most of the results of this section appear in~\cite{gruica2020common,antrobus2019maximal}. Note that the problem 
of computing the asymptotic density of
$\F_{q^m}$-linear
rank-metric codes as $m \to +\infty$
is different from the one studied in this section. We refer the reader to~\cite{neri2018genericity} for further details.

In order to simplify arguments in the sequel, we introduce the following notation and terminology.

\begin{definition} \label{def:delta}
For integers $1 \le d \le n$ and $2 \le n \le m$, let $k=m(n-d+1)$ and
\begin{align*}
\delta_q(n \times m, d)= \frac{|\{\mC \le \mat \mid \dim(\mC)=k, \; \drk(\mC) = d\}|}{\qbin{mn}{k}{q}}
\end{align*}
denote the \textbf{density} (\textbf{function}) of MRD codes in $\mat$ with minimum distance at least $d$ among all $k$-dimensional codes in $\mat$. Their 
\textbf{asymptotic density} is 
$\lim_{q \to + \infty} \delta_q(n \times m, d)$, when the limit exists.
\end{definition}

In this section, $n$, $m$, and $d$
denote integers as in Definition~\ref{def:delta}. When writing $q \to+\infty$ or $m \to+\infty$, the other parameters are treated as constants.
 
\begin{definition}
If the asymptotic density for $q \to +\infty$ is 0, then we say that the MRD codes we consider are \textbf{sparse}. 
If instead the asymptotic density is 1, we say that they are \textbf{dense}. 
\end{definition}

The problem of proving whether MRD codes are dense or not has been studied through four main approaches, one of which is the main subject of this section. We start by briefly recalling the results obtained with the other three approaches. 

In~\cite{byrne2020partition}, the authors developed an approach to study asymptotic enumeration problems in coding theory, based on the notion of a \emph{partition-balanced} family of codes. By applying the machinery to the problem of estimating the number of MRD codes, one obtains the following result.

\begin{theorem} \label{thm:ravby} If $d \ge 2$, then
\begin{align*}
    \limsup_{q \to +\infty}\, \delta_q(n \times m, d) \le 1/2.
\end{align*}
\end{theorem}

Theorem~\ref{thm:ravby} shows that MRD codes are never dense for $d \ge 2$. A sharper upper bound is obtained 
in~\cite{antrobus2019maximal} using a generalization of the approach taken in Example~\ref{ex:specs} based on the theory of spectrum-free matrices. It reads as follows.

\begin{theorem} \label{thm:antrgl}
We have
\begin{align*}
    \limsup_{q \to +\infty} \, \delta_{q}(n \times m, d) \le \left(\sum_{i=0}^{m}\frac{(-1)^i}{i!}\right)^{(d-1)(n-d+1)}.
\end{align*}
\end{theorem}

In \cite{antrobus2019maximal}, it is also shown that the bound of
Theorem~\ref{thm:antrgl} is sharp
for $d=n=2$ and any $m \ge 2$.

Finally, in \cite{gluesing2020sparseness} the exact number of MRD codes with the parameters $m=n=d=3$ is computed, showing that these codes are sparse. The approach of~\cite{gluesing2020sparseness} is based on the connection between full-rank square MRD codes and semifields.

In this section we will show a fourth approach in the attempt of answering the question of whether or not MRD codes are sparse or dense. In order to address this problem, we look at MRD codes from two different viewpoints. We need the following definition.

\begin{definition} \label{def:ballrk}
The \textbf{ball} of radius $0 \le r \le n$ in $\mat$ is the set of matrices $X \in \mat$ with $\rk(X) \le r$.
It is well-known that its size 
is
\begin{equation} \label{def:ball}
    \bbq{n \times m, r}= \sum_{i=0}^r\qbin{n}{i}{q} \prod_{j=0}^{i-1}(q^m-q^j).
\end{equation}
\end{definition}

\begin{remark} \label{rem:ccandball}
Let $m \ge n \ge 2$ and $1 \le d\le n$ be integers. The set of MRD codes in $\mat$ of minimum distance $d$ can be characterized in the following two ways.
\begin{itemize}
\item[(i)] Consider the collection $\mU$ of subspaces $U \le \F_q^n$ with $\dim(U)=d-1$. For $U \in \mU$, denote by $\smash{\mat(U)}$ the set of all matrices 
whose column space is contained in $U$. It is not hard to see that $\smash{\mat(U)}$ is a linear space of dimension $m(d-1)$ for all $U \in \mU$. Moreover, let
$\mA=\{\mat(U) \mid U \in \mU\}$.
Then the common complements of the spaces in $\mA$ are exactly the rank-metric codes $\smash{\mC \le \mat}$ with
$\smash{\drk(\mC)=d}$ and
$\dim(\mC)=m(n-d+1)$, i.e., the MRD codes of dimension $k$ and minimum distance $d$. We clearly have
\begin{align*}
    |\mA| = |\mU| = \qbin{n}{d-1}{q}.
\end{align*}
\item[(ii)] On the other hand, MRD codes of minimum distance $d$ in $\mat$ can be seen as those $m(n-d+1)$-dimensional spaces in $\mat$ which do not contain any matrix of rank smaller or equal to $d-1$. In particular, they are the $m(n-d+1)$-dimensional spaces avoiding the ball of radius $d-1$ defined in Definition~\ref{def:ballrk}.
\end{itemize}
\end{remark}

A convenient way of approaching the problem of estimating the number of common complements of a collection of subspaces, and also the number of spaces avoiding a certain set, is by translating it into the question of counting the number of isolated vertices in bipartite graphs exhibiting certain regularity properties, as will be explained shortly. The machinery used in this section is based on the approach developed in \cite{gruica2020common}. We start by introducing the needed definitions and notions.

\begin{definition}
A \textbf{bipartite graph} is a 3-tuple $\mB=(\mA,\mW,\mE)$, where $\mA$, $\mW$ are finite non-empty sets and $\mE \subseteq \mA \times \mW$. The elements of $\mA \cup \mW$ are the \textbf{vertices} of the graph and the elements of $\mE$ are the \textbf{edges}. We say that a vertex~$W \in \mW$ is
\textbf{isolated} if there is no $A \in \mA$ with $(A,W) \in \mE$.
\end{definition}

Note that a lower bound on the number of non-isolated vertices is equivalent to having an upper bound on the number of isolated vertices in $\mW$ (and vice-versa), as the two sets are complements of each other. The graphs for which we give bounds are regular with respect to an association as explained next.

\begin{definition} \label{def:assoc}
Let $\mA$ be a finite non-empty set and let $r \ge 0$ be an integer. An \textbf{association} on $\mA$ of \textbf{magnitude} $r$ is a function 
$\alpha: \mA \times \mA \to \{0,...,r\}$ satisfying the following:
\begin{itemize}
\item[(i)] $\alpha(A,A)=r$ for all $A \in \mA$;
\item[(ii)] $\alpha(A,A')=\alpha(A',A)$ for all $A,A' \in \mA$.
\end{itemize}
\end{definition}

\begin{definition}
Let $\mB=(\mA,\mW,\mE)$ be a finite bipartite graph and let $\alpha$ be an association on~$\mA$ of magnitude $r$.  We say that $\mB$ is \textbf{$\alpha$-regular} if for all  $(A,A') \in \mA \times \mA$ the number of vertices $W \in \mW$ with $(A,W) \in \mE$ and 
$(A',W) \in \mE$ only depends on $\alpha(A,A')$. If this is the case, we denote this number by~$\mW_\ell(\alpha)$, where $\ell=\alpha(A,A')$. More formally, we let
\begin{align*}
    \mW_\ell(\alpha) = |\{W \in \mW \mid (A,W), (A',W) \in \mE\}|
\end{align*}
for any $A,A' \in \mA$ with $\ell = \alpha(A,A')$.
\end{definition}

Combining the notion of an association and the Cauchy-Schwarz Inequality, one can prove the following lemma.

\begin{lemma} \label{lem:upperbound}
Let $\mB=(\mA,\mW,\mE)$ be a finite bipartite $\alpha$-regular graph, where $\alpha$ is an association on~$\mA$ of magnitude~$r$. Let $\mF \subseteq \mW$ be the collection of non-isolated vertices of $\mW$. If
$\mW_r(\alpha) >0$, then
$$|\mF| \ge  \frac{\mW_r(\alpha)^2 \, |\mA|^2}{\sum_{\ell=0}^r  \mW_\ell(\alpha) \, |\alpha^{-1}(\ell)|}.$$
\end{lemma}
\begin{proof}
Define the set $\mathbf{S}=\{(A,A',W) \in \mA^2 \times \mW \mid (A,W) \in \mE, \, (A',W) \in \mE\}$.
Since all the vertices in $\mW \setminus \mF$ are isolated, we have 
\begin{equation}\label{ntt1}
    |\mF| \cdot |\mathbf{S}| = |\mF| \, \sum_{W \in \mF} |\{A \in \mA \mid (A,W) \in \mE\}|^2 \ge 
\left( \sum_{W \in \mF} |\{A \in \mA \mid (A,W) \in \mE\}| \right)^2,
\end{equation}
where the last bound follows immediately from the Cauchy-Schwarz Inequality.
We have
\begin{equation}\label{ntt2}
   \sum_{W \in \mF} |\{A \in \mA \mid (A,W) \in \mE\}| = \sum_{A \in \mA} |\{W \in \mW \mid (A,W) \in \mE\}|
   = \mW_r(\alpha) \, |\mA|.
\end{equation}
Combining \eqref{ntt1} with \eqref{ntt2} one obtains
\begin{equation} \label{ntt3}
    |\mF| \cdot |\mathbf{S}| \ge \mW_r(\alpha)^2 \, |\mA|^2.
\end{equation}
Moreover, by the definition of association,
\begin{align} \label{ntt4}
    |\mathbf{S}| &=  \sum_{\ell=0}^r \; \sum_{\substack{(A,A') \in \mV^2, \\ \alpha(A,A')=\ell}} |\{W \in \mW \mid (A,W) \in \mE, \, (A',W) \in \mE|\} \nonumber \\
    &= \sum_{\ell=0}^r \mW_\ell(\alpha) \cdot 
    |\{(A,A') \in \mA^2 \mid \alpha(A,A')=\ell\}| \nonumber \\
    &= \sum_{\ell=0}^r  \mW_\ell(\alpha) \, |\alpha^{-1}(\ell)|.
\end{align}
Since $\mW_r(\alpha) >0$,
we have $|\mathbf{S}| \neq 0$. Therefore to conclude the proof it suffices to combine~\eqref{ntt3} with~\eqref{ntt4}.
\end{proof}

We now apply Lemma~\ref{lem:upperbound} to derive an upper bound for the number of MRD codes. We do this in two different ways, first by looking at MRD codes as common complements, as explained in Remark~\ref{rem:ccandball} (i). We need the following two results. We omit the proofs but they can be found in~\cite{gruica2020common}. 

\begin{lemma} \label{lem:nu}
The number of $k$-spaces
$W \le \mat$
intersecting $(mn-k)$-spaces $A,B \le \mat$ only depends on $\ell = \dim(A \cap B)$. We denote this number by
$\nu_q(mn,k,\ell)$. Furthermore, for $mn-2k \le \ell \le mn-k$ we have
\begin{align*}
    \nu_q(mn,k,\ell) =
    \qbin{mn}{k}{q} - 2q^{k(mn-k)} + q^{(2k-mn+\ell)(mn-k)}\prod_{i=\ell}^{mn-k-1}(q^{mn-k}-q^i).
\end{align*}
\end{lemma}

The previous lemma was derived passing through classical  methods from the theory of \textit{critical problems} in combinatorial geometry, proposed
by Crapo and Rota in~\cite{crapo1970foundations}.

\begin{lemma} \label{lem:theta}
For non-negative integers
$u$ and $i$ with $u \le n$ and $2u-n \le i \le u$, we denote by $\theta_q(n,u,i)$ the number of pairs $(U,U')$
of $u$-dimensional spaces
 $U,U' \le \F_q^n$ with the property that 
$\dim(U \cap U')=i$. We have 
$$\theta_q(n,u,i) =  \sum_{j=i}^u (-1)^{j-i} q^{\binom{j-i}{2}} \, \qbin{n}{i}{q}
\, \qbin{n-i}{j-i}{q} \, \qbin{n-j}{u-j}{q}^2.$$
\end{lemma}

We are now ready to apply the bound of Lemma~\ref{lem:upperbound} to the setting of Remark~\ref{rem:ccandball} (i) and get a first bound on the number of MRD codes.

\begin{theorem} \label{thm:boundcc}
Let $2 \le n \le m$ and $1 \le d \le n$ be integers and fix $k=m(n-d+1)$. Let $\mA$ be defined as in Remark~\ref{rem:ccandball} (i). We have the following bound on the number of MRD codes.
\begin{align*}
    |\{\mC \le \mat \mid \dim(\mC)=k, \; \drk(\mC) = d\}| \le \qbin{mn}{k}{q}- \frac{ \nu_q(mn,k,m(d-1))^2 \, \qbin{n}{d-1}{q}^2 }{\displaystyle\sum_{i=0}^{d-1} \nu_q(mn,k,mi) \,  \theta_q(n,d-1,i)}.
\end{align*}
In particular, we have
\begin{align*}
    \delta_q(n \times m, d) &\le 1 - \frac{ \nu_q(mn,k,m(d-1))^2 \, \qbin{n}{d-1}{q}^2 }{\qbin{mn}{k}{q}\,\displaystyle\sum_{i=0}^{d-1} \nu_q(mn,k,mi) \,  \theta_q(n,d-1,i)}.
\end{align*}
\end{theorem}
\begin{proof}
Let $\mA$ and $\mU$ be the sets defined as in Remark~\ref{rem:ccandball} (i).
We apply Lemma~\ref{lem:upperbound} to the bipartite graph $\mB=(\mA,\mW,\mE)$, where $\mW$ is the collection of $k$-subspaces of $\mat$ and for $A \in \mA$ and $\mC \in \mW$ we let $(A,\mC) \in \mE$ if $A$ and $\mC$ intersect non-trivially. We define
an association~$\alpha$ of magnitude $mn-k=m(d-1)$ on $\mA$ by setting $\alpha(A,A')=\dim(A \cap A')$ for all $A,\,A' \in \mA$. By Lemma~\ref{lem:nu}, the graph $\mB$ is $\alpha$-regular with $\mW_{\ell}(\alpha)=\nu_q(mn,k,\ell)$
for all $\ell \in \{0,...,m(d-1)\}$. Note that we have
$|\alpha^{-1}(\ell)| =  |\{(A,A') \in \mA' \mid \dim(A \cap A')=\ell\}|$ for all $\ell \in \{0,...,m(d-1)\}$.
Furthermore, for $U,U' \in \mU$ we have $$\dim(\mat(U)\cap\mat(U'))=\dim(\mat(U\cap U')) = mi$$ for $i= \dim(U \cap U') \in \{0,1,...,d-1\}$. Therefore, by Lemma~\ref{lem:theta}, it holds that
\begin{align*}
    \sum_{\ell=0}^{mn-k}  \nu_q(mn,k,\ell) \; |\{(A,A') \in \mA^2 \mid \dim(A\cap A')=\ell\}| 
    = \sum_{i=0}^{d-1} \nu_q(mn,k,mi) \, \theta_q(n,d-1,i),
\end{align*} 
which means we have all the needed quantities to apply Lemma~\ref{lem:upperbound}. We get
\begin{align*}
    |\{\mC \le \mat \mid \dim(\mC)=k, \; \drk(\mC) \le d-1\}| \ge \frac{ \nu_q(mn,k,m(d-1))^2 \, \qbin{n}{d-1}{q}^2 }{\displaystyle\sum_{i=0}^{d-1} \nu_q(mn,k,mi) \,  \theta_q(n,d-1,i)},
\end{align*}
which in turn yields the desired result.
\end{proof}

The other bound we obtain using Lemma~\ref{lem:upperbound} is by looking at MRD codes with minimum distance $d$ as the spaces which do not contain any matrix of rank smaller than or equal to $d-1$, as explained in Remark~\ref{rem:ccandball} (ii).

\begin{theorem} \label{thm:boundball}
Let $2 \le n \le m$ and $1 \le d \le n$ be integers and fix $k=m(n-d+1)$. We have 
\begin{align*}
    |\{\mC \le \mat \mid \dim(\mC)=k, \; \drk(\mC) = d\}| &\le \qbin{mn}{k}{q} - \displaystyle\frac{\mathfrak{B}\,\qbin{mn-1}{k-1}{q}^2}{\qbin{mn-1}{k-1}{q}+\displaystyle\left(\mathfrak{B}-1\right)\qbin{mn-2}{k-2}{q}},
\end{align*}
where $ \mathfrak{B}=(\bbq{n \times m,d-1}-1)/(q-1).$
In particular,
\begin{align*}
    \delta_q(n \times m, d) &\le 1 - \displaystyle\frac{\mathfrak{B}\,\qbin{mn-1}{k-1}{q}^2}{\qbin{mn}{k}{q}\left(\qbin{mn-1}{k-1}{q}+\displaystyle\left(\mathfrak{B}-1\right)\qbin{mn-2}{k-2}{q}\right)}.
\end{align*}
\end{theorem}
\begin{proof}
The proof is similar to the proof of Theorem~\ref{thm:boundcc}. This time we consider the bipartite graph $\mB=(\mA,\mW,\mE)$, where $\mA$ is the set of matrices in the ball in $\mat$ of radius $d-1$ (up to multiples), $\mW$ is the collection of $k$-subspaces in $\mat$ and $(A,\mC) \in \mE$ if $A \in \mC$. We define
an association~$\alpha$ on $\mA$ by setting $\alpha(A,A')=|\{A,A'\}| \in \{1,2\}$ for all $A,\,A' \in \mA$. It is not hard to see that we have
\begin{align*}
    &|\mA| = \mathfrak{B}, \quad |\mW_1(\alpha)| = \qbin{mn-1}{k-1}{q}, \quad  |\mW_2(\alpha)| = \qbin{mn-2}{k-2}{q}, \\[0.2cm]
    &|\alpha^{-1}(1)| = \mathfrak{B}, \quad |\alpha^{-1}(2)| = \mathfrak{B}(\mathfrak{B}-1).
\end{align*}
Combining Lemma~\ref{lem:upperbound} with the above identities yields the desired bound.
\end{proof}

We have now established two different bounds for the number of MRD codes with certain parameters and for $\delta_q(n \times m,d)$ (Theorem~\ref{thm:boundcc} and Theorem~\ref{thm:boundball}). It turns out that both bounds on $\delta_q(n \times m,d)$ give the same asymptotic estimate stated in the next theorem. Even though it is not obvious how to prove this for the bound in Theorem~\ref{thm:boundcc}, using the concept of an \textit{asymptotic partial spread} (see~\cite{gruica2020common}), the estimate can be derived. Furthermore, experimental results indicate that the bound derived from Theorem~\ref{thm:boundcc} is in general better than the one of Theorem~\ref{thm:boundball}.

\begin{theorem} \label{thm:sparseness}
We have
$$\delta_q(n \times m, d) \in O\left(q^{-(d-1)(n-d+1)+1}\right) \quad \mbox{as $q \to +\infty$}.$$
In particular, 
$$\lim_{q \to +\infty} \delta_q(n \times m, d) =
\begin{cases} 
1 &\mbox{if $d=1$,} \\
\sum_{i=0}^m \frac{(-1)^i}{i!} & \mbox{if $n=d=2$,} \\
0 & \mbox{otherwise}. 
\end{cases}$$
\end{theorem}

Theorem~\ref{thm:sparseness} computes the asymptotic density of MRD codes as $q \to +\infty$ for all parameter sets, showing that they are (very) sparse whenever $n \ge 3$ and $d \ge 2$. 

Another interesting question that concerns the proportion of MRD codes within the set of codes of the same dimension is the following: how does the density function $\delta_q(n \times m,d)$ of MRD codes behave, for fixed $q$ and $n \ge d$, as $m \to +\infty$?
We start by surveying the current literature on this question. The analogue of Theorem~\ref{thm:ravby} for $m \to +\infty$ reads as follows.

\begin{theorem}
For all $d \ge 2$ we have
\begin{align*}
    \limsup_{m \to +\infty}\, \delta_q(n \times m,d) \le \frac{(q-1)(q-2)+1}{2(q-1)^2} \le \frac{1}{2}.
\end{align*}
\end{theorem}

A bound which is even sharper than the one of~\cite{byrne2020partition} is given in~\cite{antrobus2019maximal}. We restated the bound in the language of this paper using the asymptotic estimate of the $q$-ary coefficient 
\begin{align} \label{eq:qbinm}
\qbin{mn}{m(n-d+1)}{q} \sim q^{m^2(d-1)(n-d+1)}\,\prod_{i=1}^{\infty}\displaystyle\left(\frac{q^i}{q^i-1}\right) \quad \mbox{as $m\to +\infty$}.
\end{align}
 Their result combined with~\eqref{eq:qbinm} reads as follows.

\begin{theorem} \label{heid:m}
For all $d \ge 2$ we have
\begin{align*}
    \limsup_{m \to +\infty} \, \delta_{q}(n \times m, d) \le \prod_{i=1}^{\infty} \left(1-\frac{1}{q^i}\right)^{q(d-1)(n-d+1)+1}.
\end{align*}
\end{theorem}

Since the approach based on looking at MRD codes as common complements of a collection of linear subspaces performs better than the methods developed in~\cite{antrobus2019maximal,byrne2020partition} for increasing field size and also a bit better than the approach of Remark (ii)~\ref{rem:ccandball}, we compared the bound on the density in Theorem~\ref{thm:boundcc} to known results for increasing column length. Computing the asymptotic estimate of the bound in Theorem~\ref{thm:boundcc} for $m \to +\infty$ yields the following result.

\begin{theorem} \label{thm:mrdboundccm}
Let $d \ge 2$. We have 
\begin{equation*}
\limsup_{m \to +\infty} \, \delta_q(n \times m,d) \le  \frac{1}{\displaystyle\qbin{n}{d-1}{q}\left(\prod_{i=1}^{\infty}\displaystyle\left(\frac{q^i}{q^i-1}\right)-1\right)+1}.
\end{equation*}
\end{theorem}

The asymptotic upper bound for the number of MRD codes in $\mat$ of minimum distance~$d$ in Theorem~\ref{thm:mrdboundccm} is never 0 for fixed $n,d$ and $q$. 

We want to compare the upper bounds for the asymptotic density of MRD codes as \smash{$m \to +\infty$} in Theorem~\ref{heid:m} and Theorem~\ref{thm:mrdboundccm}. For this, first note that the bound in Theorem~\ref{heid:m} can be rewritten as 
\begin{align*}
    \displaystyle\displaystyle\frac{1}{\displaystyle\prod_{i=1}^{\infty}\displaystyle\left(\frac{q^i}{q^i-1}\right)^{q(d-1)(n-d+1)+1}}.
\end{align*}

\begin{remark}
The infinite product that shows up in both expressions is closely related to the Euler $\phi$ function \smash{$\phi:(-1,1) \to \R$} defined by $\phi(x) = \prod_{i=1}^{\infty}(1-x^i)$
for all $x \in (-1,1)$; see~\cite[Section 14]{apostol2013introduction}. More precisely, we have
\[
\prod_{i=1}^{\infty}\displaystyle\left(\frac{q^i}{q^i-1}\right) = \frac{1}{\phi(1/q)}.
\]
\end{remark}

By Euler's Pentagonal Theorem~\cite[Theorem 14.3]{apostol2013introduction}, which expresses the infinite product $\phi(x)$ as an infinite sum,
$$\phi(x) = 1+\sum _{k=1}^{\infty }(-1)^{k}\left(x^{k(3k+1)/2}+x^{k(3k-1)/2}\right) = 1-x-x^{2}+x^{5}+x^{7}-x^{{12}}-x^{{15}}+\ldots.$$
one can prove that
\begin{equation*}
\lim_{q \to +\infty} \frac{\displaystyle\prod_{i=1}^{\infty}\displaystyle\left(\frac{q^i}{q^i-1}\right)^{q(d-1)(n-d+1)+1}}{\qbin{n}{d-1}{q}\left(\displaystyle\prod_{i=1}^{\infty}\displaystyle\left(\frac{q^i}{q^i-1}\right)-1\right)+1} = 0
\end{equation*}
for $n \ge d \ge 2$ and $n >2$.
Therefore the bound given in Theorem~\ref{thm:mrdboundccm} is sharper than the bound of Theorem~\ref{heid:m} for $q$ large.
For small values of $q$
the bound of Theorem~\ref{heid:m} is sharper than the one in Theorem~\ref{thm:mrdboundccm} in several examples. Therefore, the two bounds are in general not comparable with each other.

\bibliographystyle{amsplain}
\bibliography{ourbib}

\end{document}